\documentclass[sigconf,publish]{acmart}
\usepackage{multirow}
\usepackage{amsmath}
\usepackage{amsthm}
\usepackage{paralist}
\usepackage{appendix}
\usepackage{array}
\AtBeginDocument{%
  \providecommand\BibTeX{{%
    \normalfont B\kern-0.5em{\scshape i\kern-0.25em b}\kern-0.8em\TeX}}}

\setcopyright{none}
\renewcommand\footnotetextcopyrightpermission[1]{}

\newtheorem{theorem}{Theorem}

\begin{document}

\title{Enabling Communication via APIs for Mainframe Applications}

\titlenote{A shorter version of this work appeared at the International
Workshop on AI for Software Modernization (AISM 2026), co-located with
ASE 2026~\cite{kanvar2026pattern}.}

\author{Vini Kanvar}
\affiliation{%
  \institution{IBM Research, India}
  \city{}
  \country{}
  }
   \email{vkanv031@in.ibm.com}

\author{Ashwin Dhinesh Kumar}
\affiliation{%
  \institution{IBM Research, USA}
  \city{}
  \country{}
  }
   \email{ashwin.dhinesh.kumar@ibm.com}
   
\author{Srikanth Tamilselvam}
\affiliation{%
  \institution{IBM Research, India}
  \city{}
  \country{}
  }
   \email{srikanth.tamilselvam@in.ibm.com}
   
\author{Keerthi Narayan Raghunath}
\affiliation{%
  \institution{IBM, India}
  \city{}
  \country{}
  }
   \email{keerag09@in.ibm.com}

\begin{abstract}

For many decades, mainframe systems have played a crucial role in enterprise computing, supporting essential applications in various industries like banking, retail, and healthcare. To leverage the power of these legacy applications and enable their reuse, there is a growing interest in using Application Programming Interfaces (APIs) to expose their data and functionalities, thus facilitating the creation of new applications. However, the process of identifying and exposing APIs for different business use cases poses significant challenges. This involves understanding the legacy code, separating dependent components, introducing new artifacts, and making changes without disrupting the functionality or compromising the key Service Level Agreements (SLAs) like Turnaround Time (TAT).

In this work, we address the challenges associated with creating APIs for legacy mainframe applications and propose a novel framework to facilitate communication for various use cases. To identify APIs, we compile a list of artifacts, including transactions, screens, control flow blocks, inter-microservice calls, business rules, and data accesses. Static analyses, such as liveness and reaching definitions, are then used to traverse the code and automatically compute API signatures, which consist of request/response fields.

To assess the effectiveness of our framework, we conducted a qualitative survey involving nine mainframe developers with an average experience of 15 years. The survey helped us identify the candidate APIs and estimate the development time required to code these APIs on a public mainframe application named GENAPP and two industry mainframe applications. The results indicated that our framework successfully identified more candidate APIs and reduced the implementation time needed for APIfication. The implementation of API signature computation is part of IBM Watsonx Code Assistant for Z Refactoring Assistant. We evaluated the correctness of the identified APIs by executing them on an IBM Z mainframe system, demonstrating the practical viability of our approach.
\end{abstract}

\keywords{COBOL, API, Screen, Signature, z Mainframe}

\maketitle
\thispagestyle{plain}
\pagestyle{plain}

\section{Introduction}
\label{sec:intro}

Mainframe systems remain a crucial component of enterprise computing in industries like banking, retail, insurance, and healthcare due to their exceptional performance, consistency, and reliability~\cite{industries-mainframe, mainframe}. However, with the growing popularity of cloud computing and microservices there is a rising trend towards "APIfication" of mainframe applications. According to McKinsey and Co in 2021, over 90\% of financial institutions either use or plan to use APIs to generate additional revenue from existing customers~\cite{mckinsey}.

Mainframe APIfication involves modernizing applications by exposing data and functionalities through application programming interfaces (APIs) with tools like zOS Connect~\cite{zosc}. This transformation offers several benefits~\cite{aci09,a11,jmcam13}:
(1) Increased Efficiency: Developers can interact more easily and efficiently with mainframe systems, leading to faster development cycles and reduced operational costs.
(2) Integration with Modern Technologies: APIfication facilitates seamless integration with newer technologies like cloud computing, artificial intelligence, and the Internet of Things (IoT).
(3) Access Control Flexibility: Exposing mainframe systems through APIs enables better access management and implementation of fine-grained security controls.
(4) Business Agility: APIfication empowers advanced analytics and enables swift responses to changing business requirements, fostering agility.
Furthermore, creating APIs for new applications opens up new business opportunities~\cite{mainframe-appmod}. However, designing the right APIs requires significant effort and expertise from experienced mainframe developers, which can be a barrier to adopting z-Cloud integration patterns~\cite{dasgupta2021ai}. To address these, innovation is essential to automate API creation, making it more efficient and reducing the cognitive load and effort.
%time needed from weeks to minutes.
%(Section~\ref{sec:study}).

% Therefore, it is natural to first define the different APIs and their data exchange format when an application is built from scratch. However, APIs have to be identified from existing code when a monolith is refactored into microservices. This is because in a monolith, only UI/Job servicing interfaces are exposed as APIs; rest of the functionality is likely to be exposed as an internal function.

\begin{figure}
    \includegraphics[width=87mm]{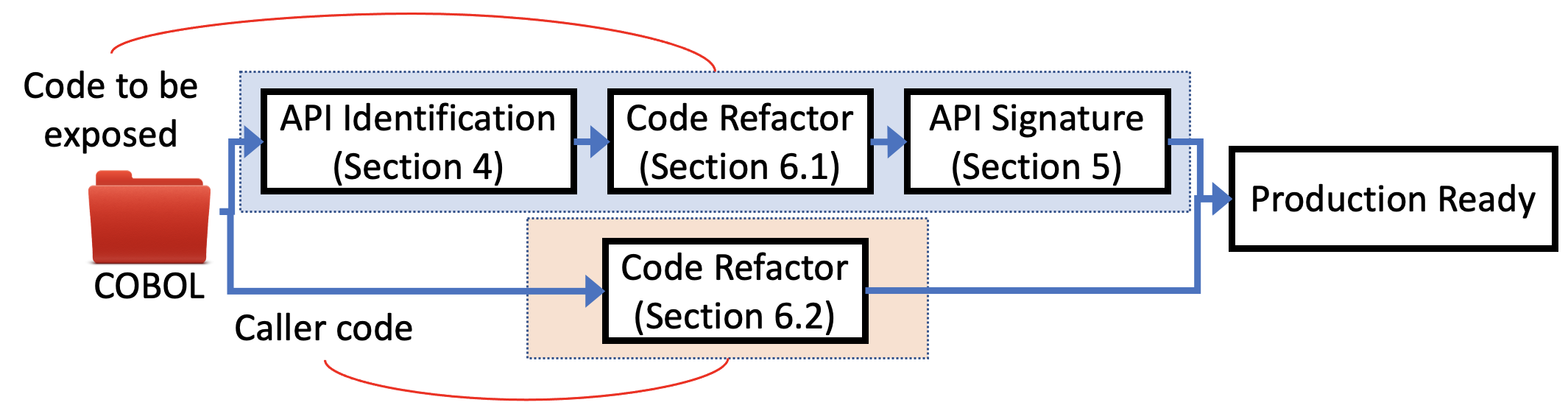} 
    \Description{}
    \caption{Mainframe APIfication Overview : Steps to APIfy a COBOL application include identification of code to be exposed, its refactoring, and computing its signature. It also includes refactoring of the caller code. Finally, it includes making the APIs production ready.}
    \label{fig:overview}
\end{figure}

Identifying APIs manually in mainframe applications presents considerable challenges compared to non-mainframe applications due to the following factors:

\begin{enumerate}
\item    Lack of Semantic Meaning: Program file and variable names in mainframe applications do not inherently convey semantic meaning, making it difficult to identify APIs based on naming conventions.

\item    Complex Data Structures: Programs in mainframe applications are often invoked through large and intricate data structures like copybooks, which are shared across multiple programs, further complicating the identification of APIs.

\item    Monolithic Codebase: Mainframe code tends to be monolithic, many times a single program encompasses multiple functions. This complexity makes it difficult to isolate APIs for specific functionalities.

\item    Non-Intuitive Code Blocks: Within a program, code blocks addressing specific functionalities are written as sections and paragraphs, and variables are passed between these blocks in a non-intuitive or non-obvious manner.

\item    Multiple Data Structures: A mainframe application's single data store can hold diverse data structures, adding complexity to API identification. For example, a VSAM file may contain various record types, including multiple "redefined" layouts, making it difficult to ascertain which code block writes each record type.

\item    UI-Backend Interaction Differences: The way mainframe screens (UI) interact/interface with the backend logic differs significantly from how web pages (GUI) interact. As more customers transition from mainframe screens to modern user interfaces, accurately capturing data exchanged between user-facing interfaces and the backend through APIs becomes crucial.
\end{enumerate}

%\subsection{Working Example: GENAPP Application}
To address the challenges mentioned earlier and facilitate Mainframe APIfication, we present a novel methodology that automates the process of API identification, computation of API signatures, and code refactoring. The final step of making the APIs production-ready which involves implementing robust error handling mechanisms, load testing etc. is left for future work, as depicted in Figure~\ref{fig:overview}. 

\begin{figure*}[t]
    \centering
    \includegraphics[width=170mm]{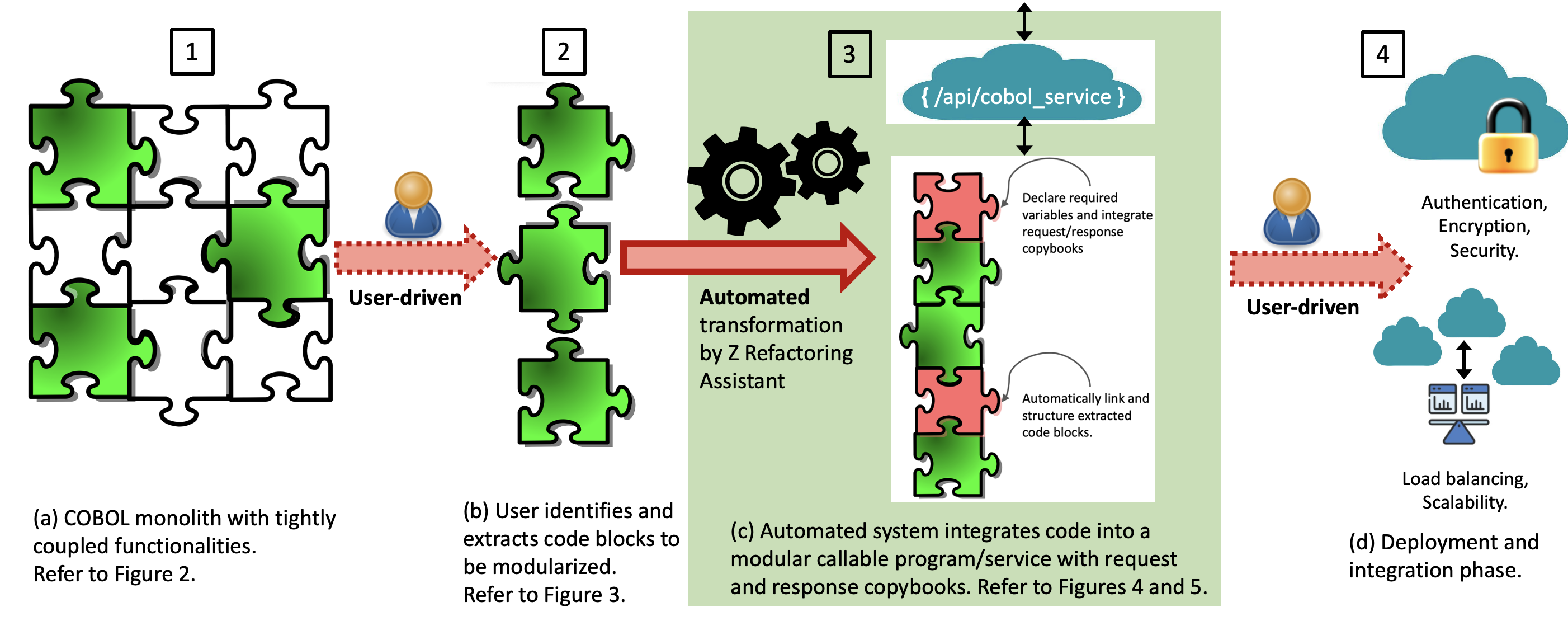}
    \caption{Overview. The green highlighted box (step 3) is our proposed system.}
    \label{fig:block}
\end{figure*}

Figure~\ref{fig:block} presents a metaphorical illustration of the system's end-to-end workflow. The highlighted green box depicts our automation, which generates a standalone, compile-able service with requests/responses. The benefits of this transformation include improved maintainability, enhanced developer interaction, seamless integration with modern systems, better access control, and support for advanced analytics.

To demonstrate the APIfication process, we utilize a publicly available mainframe COBOL application called GENAPP~\cite{genapp}. GENAPP is a General Insurance Application, supporting customer management and features such as adding, updating, and deleting policies for Housing, Motor, Endowment, and Commercial Insurance. Figure~\ref{fig:input} illustrates the "Inquire Motor Policy" feature. It involves reading variables like ENP1CNOO (customer number) and ENP1PNOO (policy number) and subsequently invoking program LGIPOL01, which, in turn, calls LGIPDB01 to access the Motor database table. Finally, it returns the Motor policy details in variables like ENP1IDAI (issue date), ENP1EDAI (expiry date), and others.

\begin{figure}[th]
    \centering
    \Description{}
    \includegraphics[width=84mm]{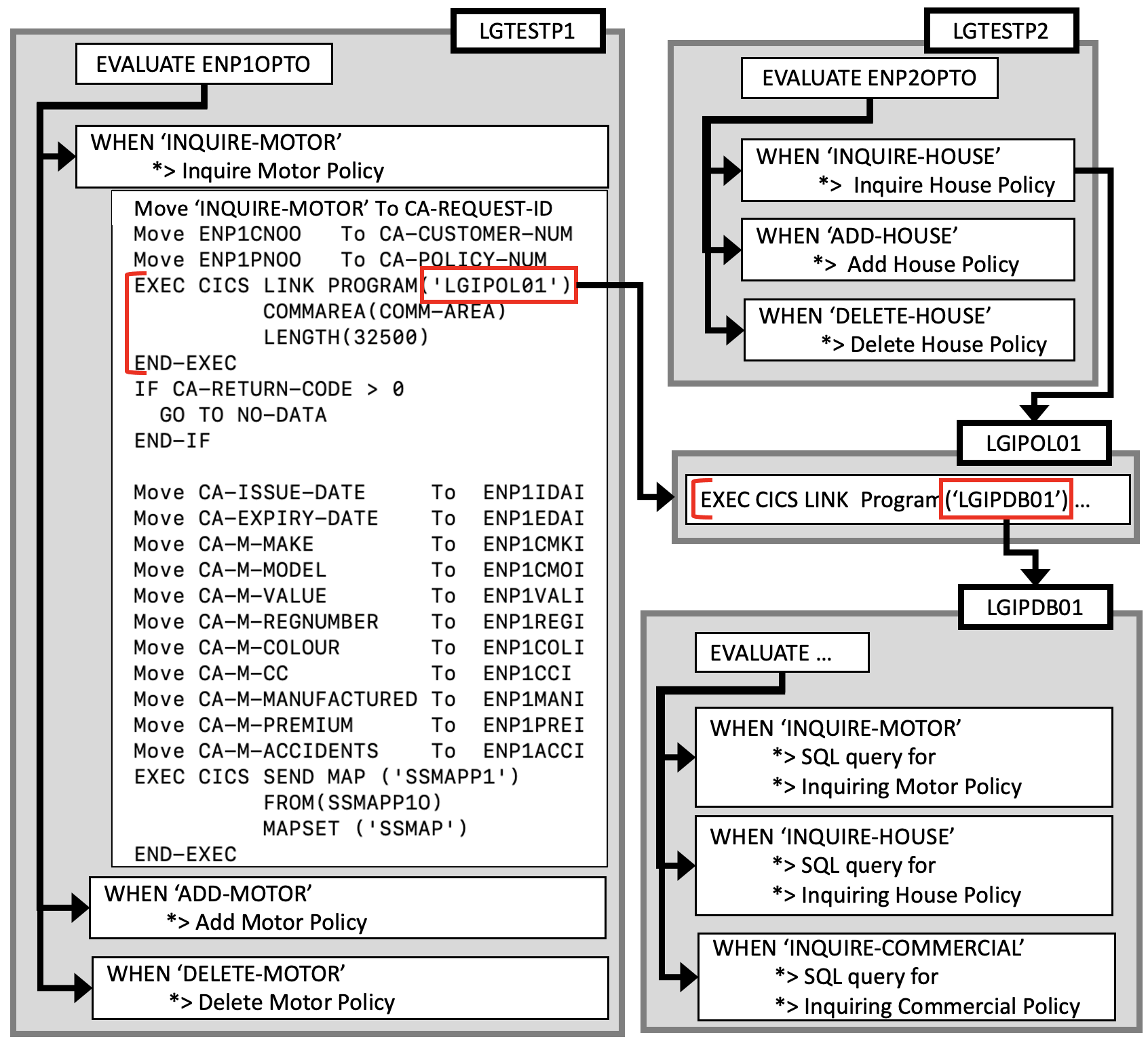}
    \caption{Monolithic GENAPP programs LGTESTP1, LGTESTP2, LGIPOL01, LGIPDB01 to inquire, add, delete on motor, house, commercial policies.}
    \label{fig:ex1}
\end{figure}

\begin{figure}[th!]
    \centering
    \Description{}
    \includegraphics[width=55mm]{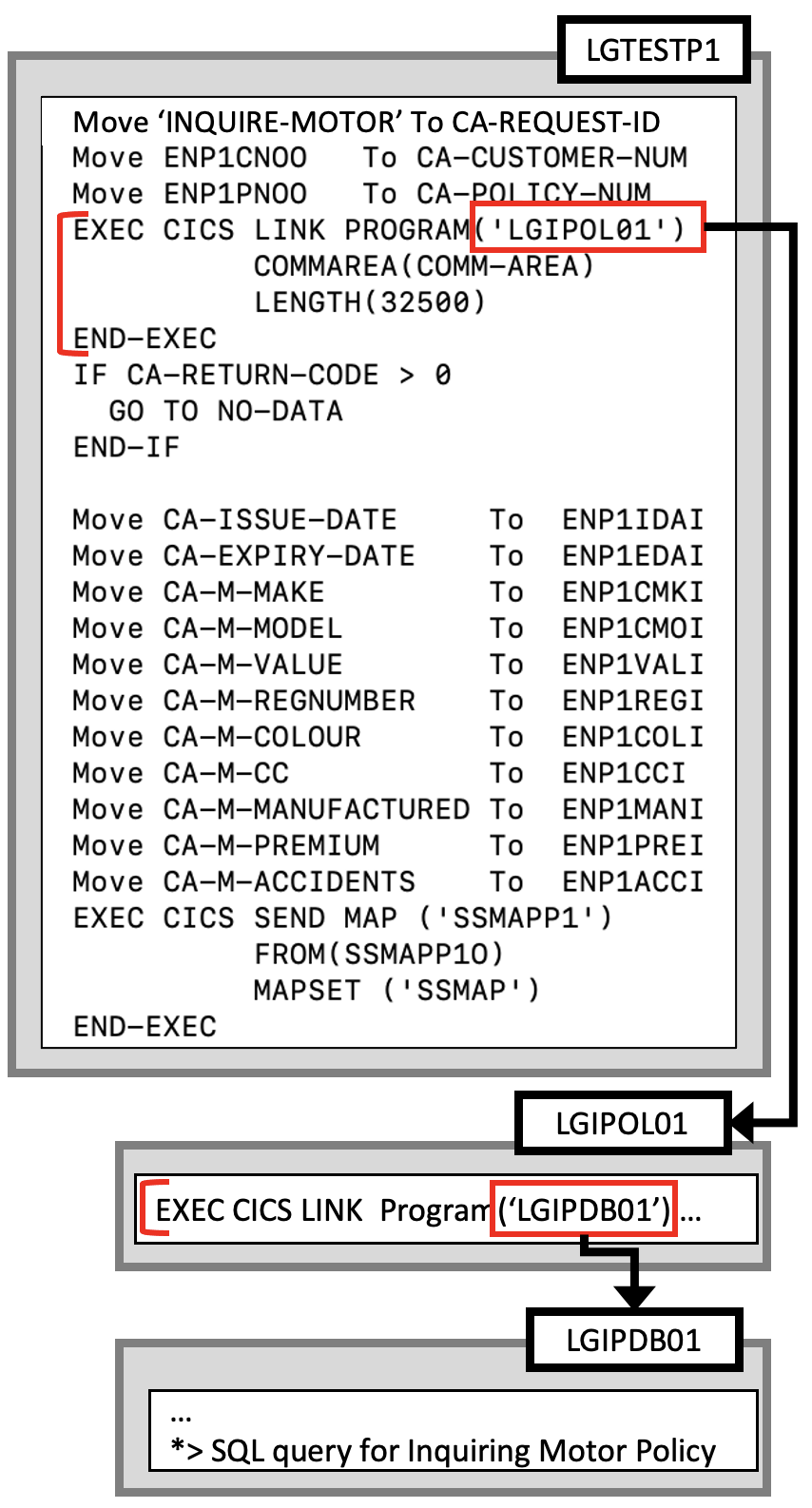}
    \caption{Code blocks relevant to inquire about motor policy to transform it into a standalone service. These blocks are extracted by the user manually.}
    \label{fig:ex2}
\end{figure}

\begin{figure}[th!]
    \centering
    \Description{}
    \includegraphics[width=70mm]{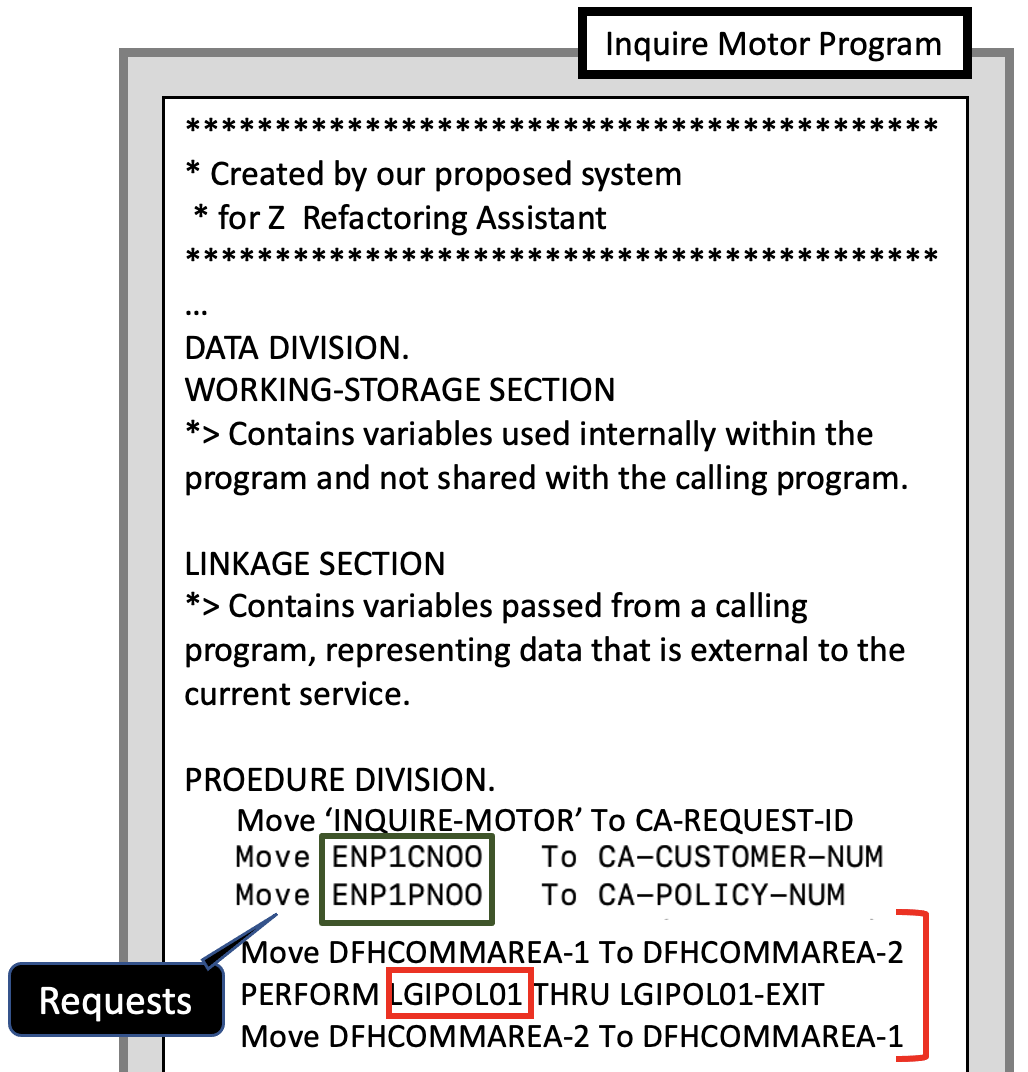}

    \includegraphics[width=69.5mm]{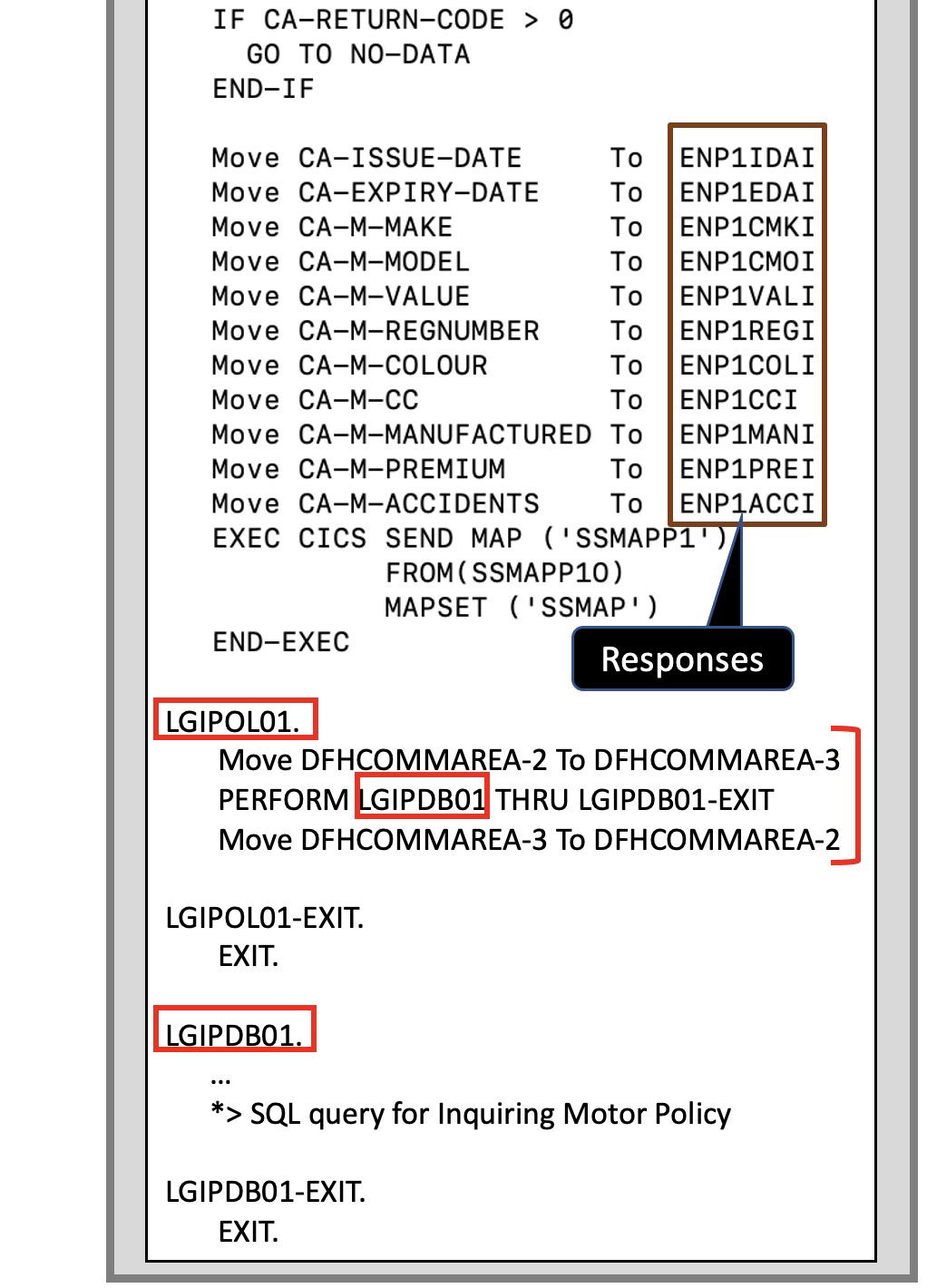}

    \includegraphics[width=70mm]{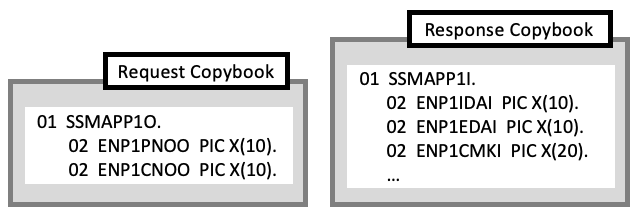}
    \caption{We generate a compile-able code and create request/response copybooks for motor policy inquiry code in Figure~\ref{fig:ex2}.}
    \label{fig:ex3}
\end{figure}

\begin{figure}[th]
\centering \Description{} \includegraphics[width=70mm]{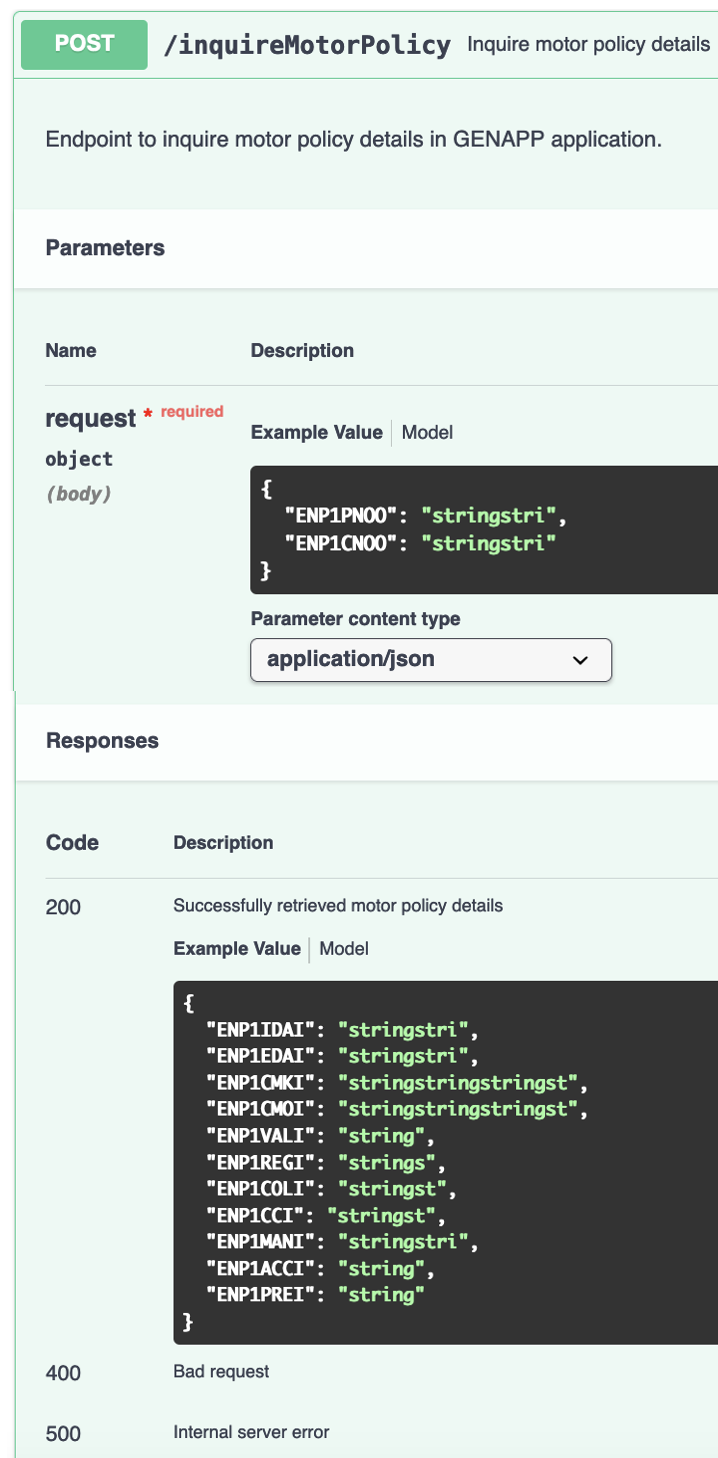}
\caption{Swagger UI for the program in Figure~\ref{fig:ex3} exposed as a service.}
\label{fig:ex4}
\end{figure}

Figure~\ref{fig:ex1} shows a spaghetti GENAPP code with multiple policies and operations handled within the same COBOL files, such as {\ttfamily LGTESTP1}, {\ttfamily LGTESTP2}, {\ttfamily LGIPOL01}, and {\ttfamily LGIPDB01}. These files mix different policy types and operations, leading to code that is difficult to navigate and maintain. Creating an API-ready module for, let us say, inquiring about a motor policy is difficult because relevant modules and variables exist in different COBOL program files. The user can leverage existing code understanding tools~\cite{ase24} or perform manual analysis to identify relevant code blocks (as shown in Figure~\ref{fig:ex2}). The figure shows three code blocks identified by the user across three programs that perform the functionality of inquiring motor policy. However, there is no automation to integrate these identified blocks, ensure their compilability, and determine request/response fields to expose the blocks as an API. Our automated {\ttfamily Z Refactoring Assistant} simplifies this process by generating a single, compilable program that integrates code blocks from one or more programs, ensuring they function as a unified business module. Additionally, our system identifies request and response variables, enabling the program to be exposed as an API and invoked as an independent module from inside as well outside the application. Manually doing this is an error-prone and time-consuming process. The output of our proposed automation for {\ttfamily Z Refactoring Assistant} is shown in Figure~\ref{fig:ex3}. 
This module can then be exposed as an API using tools like zOS Connect~\cite{zosc}, with the Swagger UI output displayed in Figure~\ref{fig:ex4}. 

Below we introduce our automated approach for identifying candidate APIs, computing the API signatures, exposing APIs, and creating the caller code:

\begin{enumerate}
\item 
    Automatic Identification of Candidate APIs: Our method involves locating the code that needs to be exposed or APIfied within the application. This identification process includes determining various artifacts acting as seeds, such as transactions, screens, control flow blocks, procedures, inter-microservice calls, business rules, data accesses, client requirements and more. For instance, in the case of GENAPP, we identify each control flow block triggered by various transactions (e.g., inquire, add, update, motor) related to Housing, Motor, Endowment, and Commercial into distinct APIs. Thus, lines 82 to 109 in Figure~\ref{fig:input} represent a candidate API for inquiring about motor policies.

\item Integrating into a Compile-able Service: 
\begin{itemize}
    \item {\em Variable Declaration and De-duplication}: We extract variables used within user-identified code blocks and declare them in the working storage and linkage sections, including memory variables, SQL cursors, DCLGen structures, and IMS PCBs. We also resolve naming conflicts and eliminate redundant variable transfers.

    \item {\em Program Call Conversion}: Inter-program calls are transformed into intra-program paragraph calls, ensuring a self-contained module. For example, EXEC CICS calls to {\ttfamily LGIPOL01} and {\ttfamily LGIPDB01} are replaced with paragraph calls after defining the required paragraphs.
\end{itemize}

\item
    Computation of API Signature: The API signature comprises the request and response fields or variables, also referred to as inputs/outputs throughout the paper. To identify these request and response fields, we perform static program analysis using techniques like data-flow analysis~\cite{dfa} and control-flow analysis~\cite{csur17}. Additionally, we propose call chain analysis for program calls within the code blocks to be APIfied. To address cyclic calls, we utilize fix point computation~\cite{dfa, csur17}. For GENAPP, as shown in Figure~\ref{fig:input}, ENP1CNOO and ENP1PNOO serve as the request fields because they are fetched by the APIfied code block; ENP1IDAI, ENP1EDAI, and others form the response fields of the API because their updated values are returned by the APIfied code block.

\item
    API Exposure and Caller Code Creation: The APIs are exposed by specifying the COBOL program to be exposed and providing the request/ response copybooks. %, integrating a Swagger file, and mapping the Swagger fields with the copybook fields. 
    Additionally, the caller code is manually prepared by invoking a communication stub to access the API. Various connect tools, such as IBM zOS Connect~\cite{zosc}, can be utilized for these steps, but the method is not limited to a particular tool.
\end{enumerate}

\begin{figure}
    \centering
    \includegraphics[width=70mm]{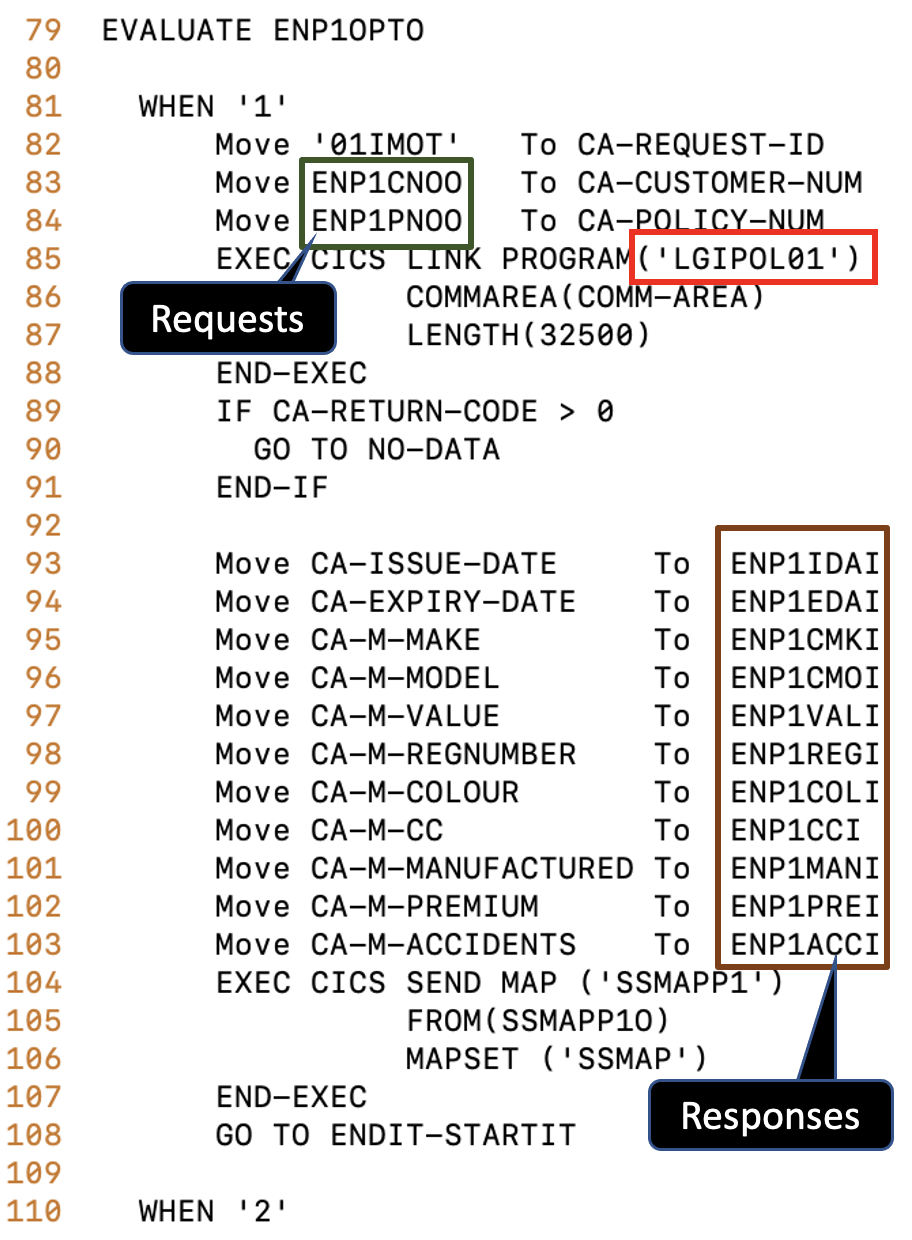}
    \Description{}
    \caption{Snippet from program LGTESTP1 of GENAPP to inquire motor policy. Lines 82 to 109 need to be exposed as an API. Requests and responses are marked for this API. Program LGIPOL01 is called from the code block.}
    \label{fig:input}
\end{figure}

Our paper makes the following main contributions:

    \begin{itemize}
    \item Automatic identification of Candidate APIs in Mainframe Applications (Section~\ref{sec:id}): Our automated approach to identify potential APIs makes the process more efficient.

    \item Automated Computation of API Signature (Section~\ref{sec:sign}): Our variants of static program analysis to automatically compute API signatures, streamlines the process and decreases the amount of manual effort needed. We formalize and prove the correctness of our analyses to establish a solid foundation for the process and accurately extract APIs.

    \item Identification of Code Refactoring Requirements for Efficient APIs (Section~\ref{sec:refactor}): We identify scenarios where code refactoring is necessary to develop efficient APIs, ensuring optimal performance and usability.

    \item {\bf Patterns for APIfication}: We analyze real-world COBOL enterprise applications to identify patterns and recommend suitable program analysis techniques for accurately computing requests/responses in Section~\ref{sec:patterns}.

    \item Demonstration of Mainframe APIfication using Industry Applications: We showcase the Mainframe APIfication process using examples from a publicly available application like GENAPP and two industry applications, employing automated tooling to illustrate its practical application.

    \item Qualitative and Quantitative Studies: We present both qualitative and quantitative studies to demonstrate the effectiveness of our technique. These studies (Sections~\ref{sec:exec} and \ref{sec:study}) provide insights into the benefits and performance improvements achieved through Mainframe APIfication. Additionally, we evaluate our tool through two client case studies, demonstrating 15\%–50\% effort reduction in COBOL refactoring and effective business logic extraction in Section~\ref{sec:client}.
    \end{itemize}

 Automatic identification of APIs, computation of API signatures, and code refactoring are in Sections~\ref{sec:id}, \ref{sec:sign}, and \ref{sec:refactor}. Experiments and qualitative study is in Section~\ref{sec:exp}. Section~\ref{sec:related} discusses related works. 
Appendix~\ref{sec:usecases} presents the business usecases of APIfication. Mathematical formalization and proofs are in Appendix~\ref{sec:formal}. %Appendix~\ref{sec:future} contains some future work.
\section{Business Usecases of APIfication}
\label{sec:usecases}

The mainframe APIfication process encompasses several common business use cases, each involving specific automated steps. These use cases and their corresponding automated steps are as follows:

\begin{itemize}
    \item 
    Create Web or Mobile UIs for Mainframe Applications:
    \begin{inparaenum}
        \item List all screens and transactions in the application.
        \item Identify the copybooks and fields used in these screens and transactions.
        \item APIfy the transactions to expose them as APIs.
        \item Create a modern UI that invokes the relevant APIs based on client requirements.
    \end{inparaenum}
    \item
    Exposing Data and Functionalities:
    \begin{inparaenum}
        \item Identify the functionality to be exposed with the assistance of subject matter experts (SMEs).
        \item Identify the code relevant to the functionality.
        \item Extract the relevant functional slice~\cite{s01} from the code.
        \item APIfy the extracted code to identify the requests and responses.
        \item Refactor the caller code to enable communication with the generated services.
    \end{inparaenum}

    \item
    Extracting Business Rules as External Services for Rule Engine Execution:
    \begin{inparaenum}
        \item Identify code blocks containing business rules using business rules extraction tools.
        \item APIfy the extracted code, as done in the previous use cases.
    \end{inparaenum}
    \item
    Refactoring Monolith Applications into Candidate Microservices:
    \begin{inparaenum}
        \item Split the monolith into functional modules with the guidance of SMEs.
        \item Identify internal services, i.e., code required by different functional modules.
        \item Identify external services.
        \item APIfy the extracted code as in the previous use cases.
    \end{inparaenum}
\end{itemize}

Using this list of business use cases, we proposed our APIfication methodology in Section~\ref{sec:id}. These automated steps enable the efficient transformation of mainframe applications into modernized, accessible APIs, offering improved functionality and integration capabilities.

\section{Automatically Identifying Candidate APIs}
\label{sec:id}
To identify APIs in mainframe applications, our methodology involves compiling a comprehensive list of artifacts, including transactions, screens, control flow blocks, inter-microservice calls, business rules, data accesses, and others present in the application.
% \begin{figure}
%     \includegraphics[width=89mm]{images/procedure-data.png} 
%     \caption{Code Snippet from program LGIPDB01 of GENAPP. It contains a procedure, which performs data access to inquire motor policy table.}
%     \label{fig:procedure-data}
% \end{figure}

\begin{enumerate}
    \item
    Transactions and Jobs: Mainframe applications typically start functionalities from transactions and jobs, which can be used to identify candidate APIs. For example, GENAPP has five main transactions - SSC1 for Customer management and SSP1, SSP2, SSP3, SSP4 for Housing, Motor, Endowment, and Commercial policies management, respectively. We begin API identification from program LGTESTP1, which is the first program called by SSP1 (Figure~\ref{fig:input}).

    \item
    Control Flow Blocks: We consider APIfying important control flow blocks, as they often contain significant functionalities related to starting transactions.

    \item
    Data Access Points: Code that performs database or file access, along with the code before and after that loads variables and handles errors, can be exposed as APIs. We offer Data APIs for both fixed and dynamic queries, enabling users to input any SQL query as a parameter to the API. This helps create a data access microservice layer. %Figure~\ref{fig:procedure-data} provides an example of a code snippet from GENAPP, where database access could be identified as an API candidate.
    %We expose Data APIs for the following cases: (i) {\em Fixed query}: Here we expose an API with a fixed SQL query, and (ii) {\em Dynamic query}: Here we provide a data interface. This enables user to input any SQL query as a parameter to the API. This requires Data Virtualisation using JDBC connection. This helps in creating a data access microservice layer.  

    \item
    Procedures as APIs: Procedures usually perform crucial business functionalities. Standalone procedures that do not call other procedures are strong candidates for conversion into APIs, representing single business functionalities. %Figure~\ref{fig:procedure-data} illustrates a code snippet from 
    For example, in GENAPP, program LGIPDB01 contains procedure GET-MOTOR-DB2-INFO, which accesses Policy and Motor tables, and could be identified as an API candidate.

    \item
    Screen APIs and Screen Modernization: Screens in mainframe applications refer to terminals for accessing the menu. By parsing screen positions specified in a text file, we can infer the fields required for each functionality. Functionalities specified in each screen can then be made into APIs, with input fields becoming requests and output fields becoming responses. Screen scraping~\cite{scraping} can be automated through APIfication, but we propose extracting relevant code and optimizing it to avoid unnecessary processes.

    \item
    Business Functionalities and Rules as APIs: We propose identifying programs carrying business rules using tools like \cite{k10,s01,se96,sdbac21,dsb20}. Our framework will extract the necessary code and host it in decision managers, such as RedHat Decision Manager\cite{redhat}.

    \item
    APIs for Inter-Microservice Communication: To enable communication between candidate microservices~\cite{mathai2021monolith}, cross-microservice function calls in a monolithic architecture need to be replaced with APIs. Our framework will identify the cross-microservices dependencies in the candidate microservices as APIs~\cite{icws22, icse23}.

    \item
    Client Requirements and External Application-Based APIs: APIs can be identified based on client demands or how the mainframe application is called externally.
\end{enumerate}

Our methodology's automated API identification approach will accelerate interface discovery and reduce the effort required to identify optimal API insertion points within mainframe applications.
\section{Automatically Finding API Signature}
\label{sec:sign}

COBOL programs commonly call other programs using copybooks, which define the data structures in COBOL. However, the number of fields in a copybook can be extensive, making it challenging to manually identify which fields will serve as requests and responses for API purposes. In this section, we present various algorithms designed to automatically identify the API signature, encompassing the request and response fields. Mathematical formalization and proofs are in Appendix~\ref{sec:formal}.
% \begin{figure}
%     \includegraphics[width=49mm]{images/output.png} 
%     \caption{Output from our tool showing requests and responses found automatically for code lines 82 to 109 of program LGTESTP1 shown in Figure~\ref{fig:input}.}
%     \label{fig:output}
% \end{figure}

\begin{figure}
    \includegraphics[width=89mm]{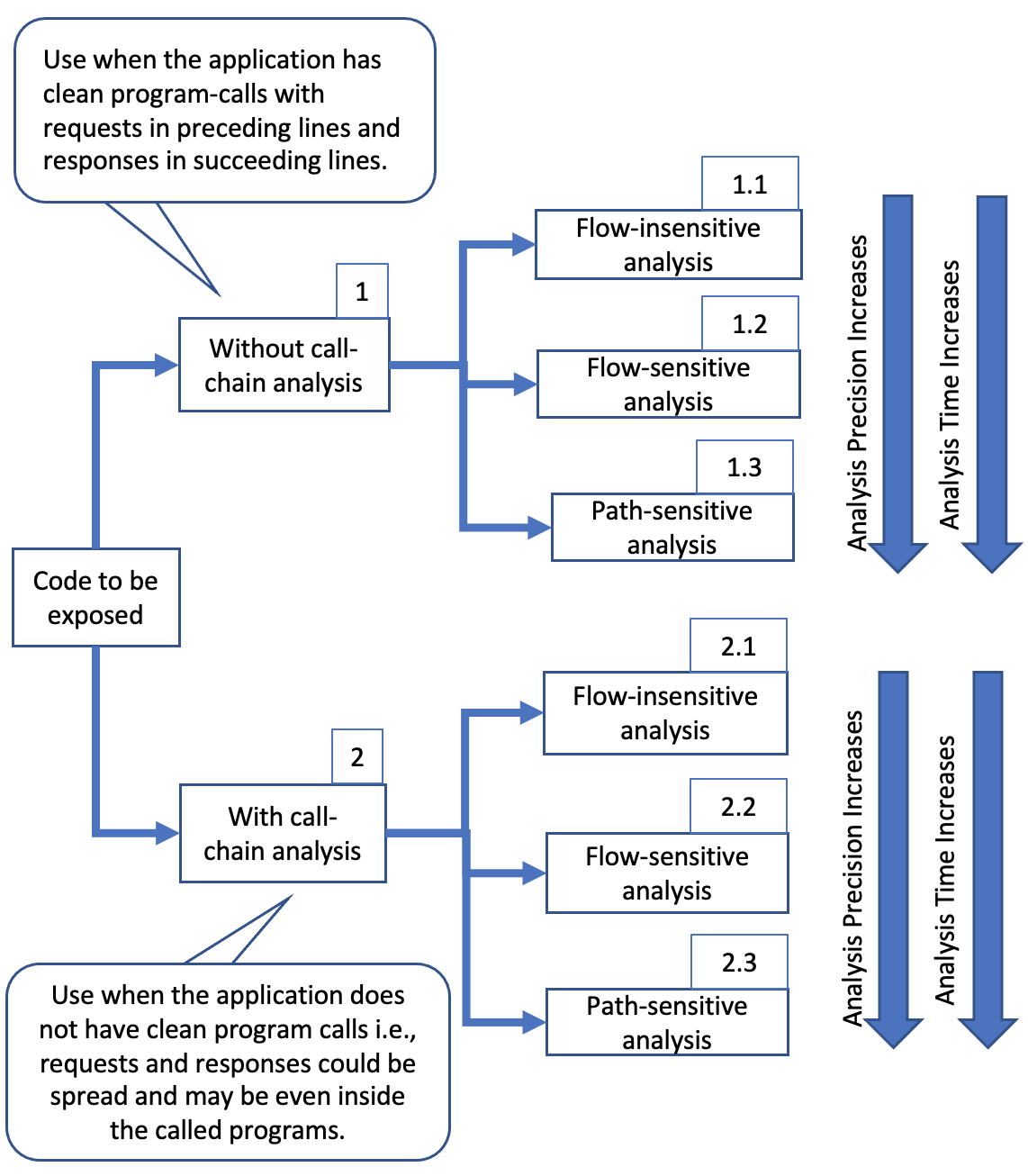} 
    \Description{}
    \caption{6 variants to compute requests/responses using combinations of with/without call chain analyses and flow-sensitive/flow-insensitive/path-sensitive analyses.}
    \label{fig:variants}
\end{figure}

\subsection{Intra-Program Analysis}
In this section, we introduce the process of defining request and response fields using static analyses. We present variants of static analyses (Sections~\ref{sec:reqres} and \ref{sec:var}) and propose a surety metric (Section~\ref{sec:surety}) for evaluating the reliability of our results.

\subsubsection{Request and Response Fields Using Static Analyses}
\label{sec:reqres}
Suppose we have code lines from $p$ to $q$ that need to be exposed as an API based on steps referred in Section~\ref{sec:id}. We propose the following techniques to determine the requests and responses of the API:

\begin{itemize}
\item
    Requests for lines $p$ to $q$: To identify request fields, we employ liveness analysis~\cite{dfa}. This analysis marks field $f$ as a request field if there exists a path from $p$ to $q$ that contains an r-value occurrence of $f$ (i.e., read of $f$) which is not preceded by an l-value occurrence of $f$ (i.e., write to $f$).

\item
    Responses for lines $p$ to $q$: To determine response fields, we use reaching definitions analysis~\cite{dfa}. This analysis marks field $f$ as a response field if there exists a path from $p$ to $q$ that contains an l-value occurrence of $f$ (i.e., write of $f$), and if there exists a path from $q$ to the application exit that contains an r-value occurrence of $f$ (i.e., read to $f$).
\end{itemize}

Unlike traditional pointer or heap abstractions~\cite{ismm17, csur17}, a COBOL field $f$ may influence the API either \emph{directly} (explicit reads/writes) or \emph{indirectly} through \texttt{REDEFINES}, i.e., shared storage layouts) and nested records that allow interchangeable use of group names and their subordinate fields. Our analysis interprets accesses to any overlapping or aliasing region as potential accesses to $f$, ensuring that request/response inference remains sound in the presence of COBOL-specific memory sharing and naming conventions.

\subsubsection{Variants of Static Analysis}
\label{sec:var}
We propose to employ three variants of static data flow analyses, namely, flow-insensitive, flow-sensitive, and path-sensitive~\cite{dfa, csur17}, to compute the requests and responses. These variants are ordered by increasing precision and analysis time. In the examples below, usages after these code snippets are not shown, so we assume the response fields are read after the API returns.

\begin{figure}
    \centering \includegraphics[width=50mm]{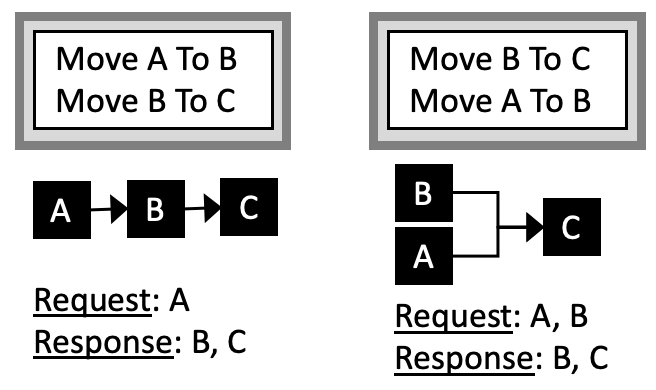} 
    \caption{Need of flow-sensitive analysis: Example to show that Request/Response change if the flow/execution order of statements is changed. The accurate Request/Response of both the programs is shown.}
    \label{fig:fs}
\end{figure}

{\bf Flow-sensitive analysis} considers the execution order, i.e., the order of reads and writes to variables, enabling precise computation of request and response fields. In Figure~\ref{fig:fs}, the left program executes {\ttfamily Move A to B} followed by {\ttfamily Move B to C}, while the right program executes {\ttfamily Move B to C} first, then {\ttfamily Move A to B}. Flow-sensitive analysis identifies request fields as those read before being written and response fields as those written in the API and read after its return.

For both programs in the figure, {\ttfamily A} is a request field because its value is read in {\ttfamily Move A to B} without a prior write within the program. Since {\ttfamily A} originates from an external source, it is correctly computed as a request. In the left program, since {\ttfamily B} originates from {\ttfamily A} (i.e., it is not read before being written), it is correctly not identified as a request. {\ttfamily B} and {\ttfamily C} are response fields because they are written and represent the final output. %Flow-sensitive analysis correctly distinguishes these roles of a request and a response based on execution order, ensuring accurate computation.

{\bf Flow-insensitive analysis}, in contrast, ignores execution order. Therefore, it cannot differentiate between the two programs in Figure~\ref{fig:fs}. Since it does not track whether a variable was assigned before being used, it follows a simple rule: all read fields are marked as request fields, and all written fields are marked as response fields. As a result, {\ttfamily A} is correctly marked as a request field, as it is read in {\ttfamily Move A to B}. However, {\ttfamily B} is incorrectly marked as a request, since flow-insensitive analysis only sees that {\ttfamily B} is read in {\ttfamily Move B to C} without considering that its value was already assigned within the program. {\ttfamily B} and {\ttfamily C} are correctly marked as response fields, as they are written to.

\begin{figure}
    \includegraphics[width=89mm]{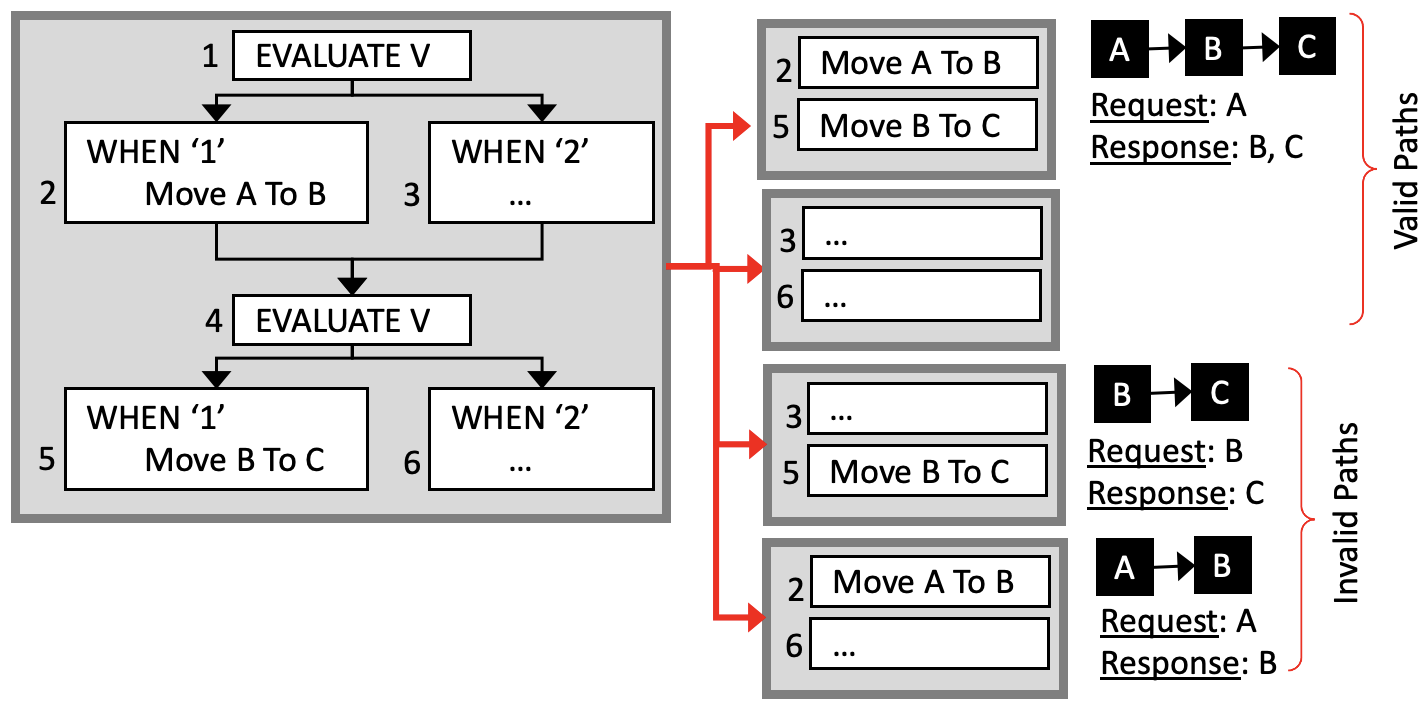} 
    \caption{Need of path-sensitive analysis: Example to show that Request/Response change if the path of statements is not honored. The Request/Response for valid and invalid paths is shown.}
    \label{fig:ps}
\end{figure}

{\bf Path-sensitive analysis} considers only valid paths from $p$ to $q$ when determining request and response fields. Here, a field is a request if it is read without a preceding write along any valid path, and a response if it is written within the API and later read after returning from the API along a valid path. Valid paths are determined based on statically computed values of conditional statements.

In Figure~\ref{fig:ps}, the program has conditional branches that determine execution paths, i.e., valid paths. When {\ttfamily V = '1'}, the valid path executes {\ttfamily MOVE A to B} followed by {\ttfamily MOVE B to C}. Here, {\ttfamily A} is correctly marked as a request because it is read without a preceding write in the program, and {\ttfamily B} and {\ttfamily C} are marked as a responses because they are written.

However, if execution does not follow a valid path, such as skipping {\ttfamily MOVE A to B} but executing {\ttfamily MOVE B to C}, then {\ttfamily B} is incorrectly marked as a request, as it appears to be read without a prior write. Such an inaccuracy arises in flow-sensitive analysis because it does not distinguish between valid and invalid paths. Path-sensitive analysis, by considering only valid execution paths, ensures precise computation of request and response fields, preventing imprecise computation seen in flow-sensitive analysis.

Flow-insensitive analysis is the least precise but is simple to implement and fast to compute. Flow-sensitive analysis is more precise but requires a more complex analysis, involving steps such as propagation of read fields backward in the control flow and fixpoint computation. To improve the time efficiency of the analysis, we compute the post-order of statements and program calls and then propagate the information~\cite{dfa}.

%Flow-insensitive analysis is the least precise but is simple to implement and fast to compute. Flow-sensitive analysis is more precise but requires a more complex analysis, involving steps such as propagation of read fields backward in the control flow and fixpoint computation. To improve the time efficiency of the analysis, we compute the post-order of statements and program calls and then propagate the information~\cite{dfa}.

In summary, we demonstrate how different variants, namely, flow-insensitive, flow-sensitive, and path-sensitive static analyses, can be leveraged to compute requests and responses. 

\subsubsection{Precision of the Analysis Results}
\label{sec:surety}
We use a boolean flag, {\em optional}, to indicate the surety of requests and responses.

\begin{itemize}
\item If field $f$ is only read along all paths from $p$ to $q$, then it is undoubtedly a request field. We mark {\em optional} flag as false.

\item If field $f$ is only written along all paths from $p$ to $q$, then it is unquestionably a response field. We mark {\em optional} flag as false.

\item If field $f$ is both read and written along some path from $p$ to $q$, then $f$ may or may not be a request or response field. We set {\em optional} flag to true. The user can manually designate field $f$ as a request, a response, or both.
\end{itemize}

Identifying API types (GET, POST, PUT, DELETE) is beyond this paper's scope. Natural Language Processing (NLP) can analyze comments and code to generate API names.

In conclusion, users have the option to select flow/path sensitive analyses based on their desired level of precision and the time constraints of the analysis.

\subsection{Inter-Module Analysis}
\label{sec:inter}
We propose two variants of traversing between modules that are called from $p$ to $q$, to determine the requests and responses for APIs, as described in Figure~\ref{fig:variants}. 

\begin{figure}
    \centering \includegraphics[width=45mm]{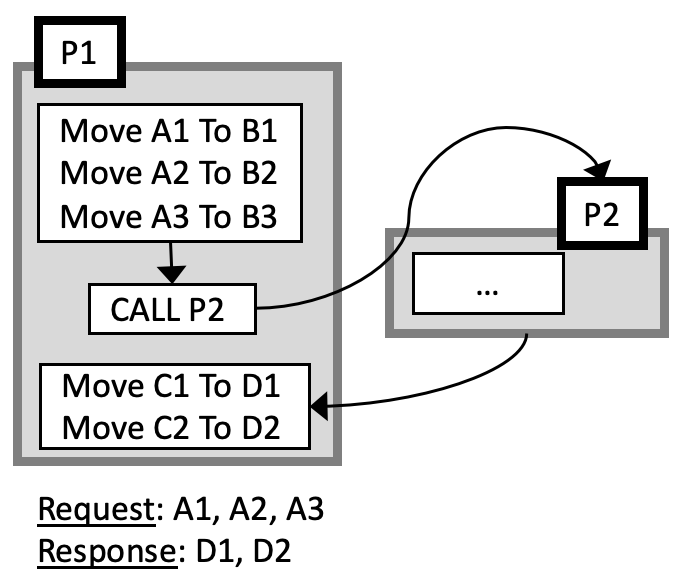} 
    \caption{A case where without call chain analysis is sufficient.}
    \label{fig:wo-call-chain}
\end{figure}

\textbf{Without call chain analysis}: Here modules that are called from $p$ to $q$ are not analyzed. We assume that the request and response fields can be computed from the code lines preceding and succeeding the module call, respectively. This approach is suitable when the monolith has well-organized module calls. As illustrated in Figure~\ref{fig:wo-call-chain} (see the example of \texttt{P1} calling \texttt{P2}), the fields \texttt{A1}, \texttt{A2}, and \texttt{A3} are moved before calling \texttt{P2}, and the fields \texttt{D1} and \texttt{D2} are populated after the call, thus identifying \texttt{A1}, \texttt{A2}, \texttt{A3} as the request and \texttt{D1}, \texttt{D2} as the response. This pattern can also be seen in GENAPP in the program \texttt{LGTESTP1} in Figure~\ref{fig:ex1}, calls  \texttt{LGIPOL01}. The request and response fields precede and succeed the program call, respectively. It is important to note that this variant may miss request and response fields present in the called module if the application is not coded in this well-organized manner.

\newcommand{\tabw}{\dimexpr\textwidth/3\relax}
\newcommand{\tabwtwo}{\dimexpr 0.22\textwidth \relax}

\begin{figure}
\footnotesize
\centering
\begin{tabular}{@{}m{\tabwtwo}||m{\tabwtwo}}
    \includegraphics[width=40mm]{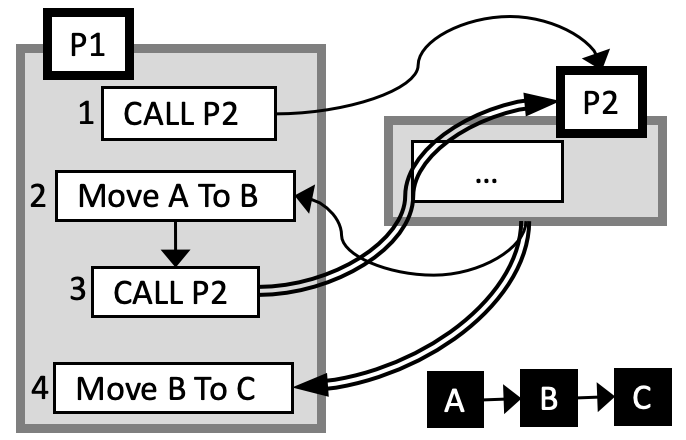} &
    \includegraphics[width=40mm]{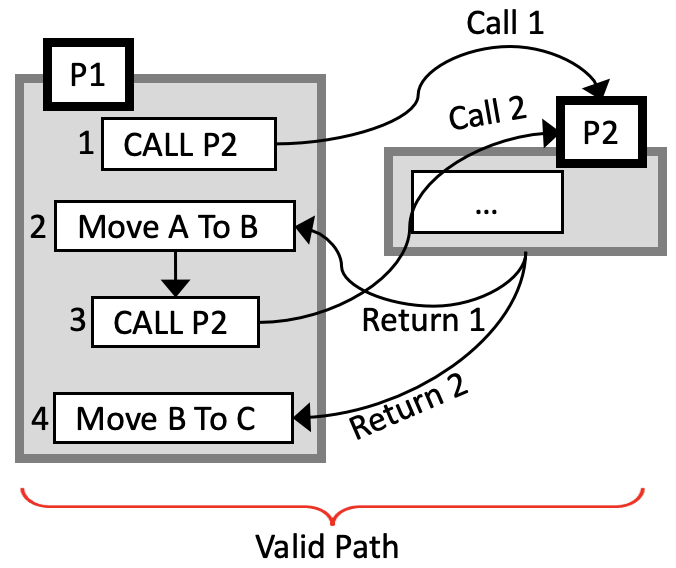} \\ \hline
    (a) Example with two calls to program P2.  &
    (b) Context-sensitive analysis distinguishes the calls and returns thereby giving correct \\ 
    \underline{Request}: A, \underline{Responses}: B, C. & 
     \underline{Request}: A, \underline{Responses}: B, C. \\ \\
    \hline \hline
    \includegraphics[width=40mm]{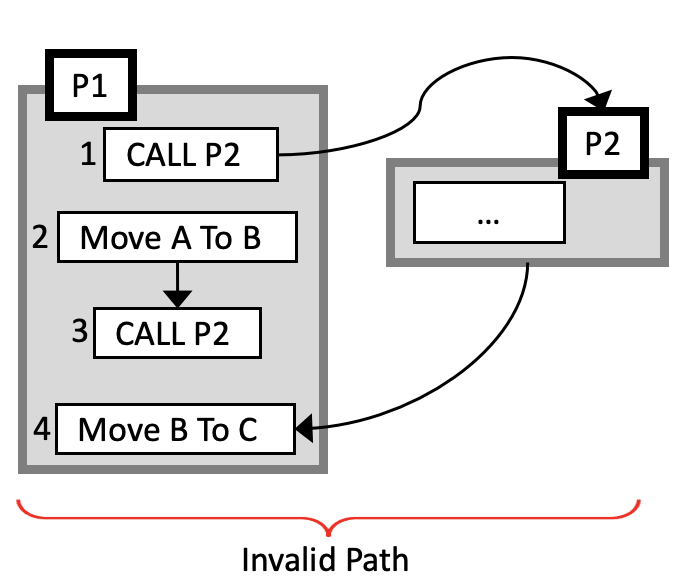} &
    \includegraphics[width=40mm]{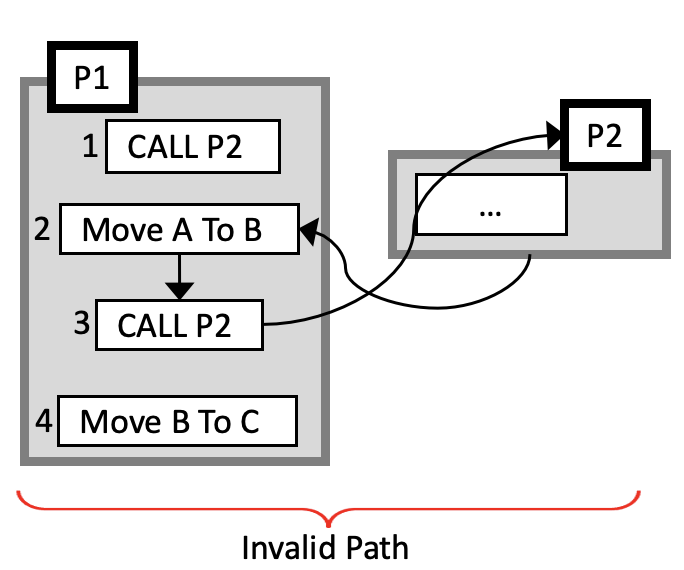} \\ \hline
    \multicolumn{2}{@{}m{\dimexpr 0.45\textwidth \relax}}
    {(c) Context-insensitive analysis covers these invalid paths, thereby giving incorrect \underline{Requests}: A, B, \underline{Responses}: B, C.}
\end{tabular}
    \caption{Need of a context-sensitive call chain analysis.}
    \label{fig:cs}
\end{figure}

\textbf{With call chain analysis}: If the monolith does not have well-organized module calls, we adopt this variant. In this approach, we analyze the code lines within the called module to determine the request and response fields. This analysis is performed recursively using fixpoint computation and can be conducted either context-sensitively or insensitively~\cite{dfa}.

A context-sensitive analysis differentiates between individual call sites, ensuring that execution paths are correctly maintained. To illustrate this, consider two calls to Program {\ttfamily P2} in Figure~\ref{fig:cs}(a) — one from statement {\ttfamily 1}, which should return to statement {\ttfamily 2}, and another from statement {\ttfamily 3}, which should return to statement {\ttfamily 4}. A context-sensitive analysis follows the valid execution path shown in Figure~\ref{fig:cs}(b), correctly identifying that {\ttfamily Move A To B} is followed by {\ttfamily Move B To C}, thereby precisely computing only {\ttfamily A} as the request.  

In contrast, a context-insensitive analysis may incorrectly merge paths, resulting in invalid execution sequences. As shown in Figure~\ref{fig:cs}(c) and Figure~\ref{fig:cs}(d), it may assume that the call from statement {\ttfamily 1} returns to statement {\ttfamily 4} or that the call from statement {\ttfamily 3} returns to statement {\ttfamily 2}. These incorrect paths disrupt the intended order of {\ttfamily Move A To B} followed by {\ttfamily Move B To C}, leading to an imprecise computation that incorrectly identifies {\ttfamily B} also as a request.  

This highlights the necessity of context-sensitive analysis in ensuring precise computation of request and response fields.

%Figure~\ref{fig:output} demonstrates a sample output (API requests and responses) in JSON format generated by our tool for lines 82 to 109 of Figure~\ref{fig:input}.

In conclusion, users have the option to select without or with call chain analysis depending on whether the application code blocks call other modules in a well organized manner. 
\subsection{Take Away: Code Patterns and Their Requests/Responses}
\label{sec:patterns}

From GENAPP and two industry applications, we distilled five common COBOL patterns influencing request/response precision. 
For each pattern, we describe the approach  that yields the most precise computation of request and response fields. Below we illustrate these patterns using code snippets from GENAPP. 
Tables~\ref{tab:patterns1} and \ref{tab:patterns2} recommend optimal configurations (FI/FS, WoCC/WCC).

\begin{table}[t]
\centering
\small
\begin{tabular}{l|p{2cm}|l|p{3.8cm}}
\toprule
\multicolumn{2}{p{2cm}|}{\textbf{Pattern}} & \textbf{Rec.} & \textbf{Description/Rationale} \\
\midrule
P1 & Sequential EVALUATE/IF-ELSE & FS & Switch statements (e.g., policy types in GENAPP). FS tracks dataflow across branches. \\ \hline
P2 & Requests-Responses around Calls & WoCC & Inputs precede, outputs follow calls. WoCC assumes localization, avoiding callee analysis. \\ \hline
P3 & Multiple Calls to Same Program & WCC & Context-sensitive traversal distinguishes call sites. \\ \hline
P4 & Chained/Transitive Assignments & FS & Tracks dataflow across group fields/assignments. \\ \hline
P5 & Nested Records/Shared Memory & FS & Handles REDEFINES aliasing via overlapping access interpretation. \\
\bottomrule
\end{tabular}
\caption{COBOL patterns and recommended analyses.}
\label{tab:patterns1}
\end{table}

\begin{figure*}[t]
\centering
\begin{tabular}{|p{5cm}|p{6cm}|p{5.4cm}|}
\hline
\textbf{COBOL Pattern} & \textbf{Key Challenge} & \textbf{Recommended Analysis Variant} \\
\hline\hline
Sequential \texttt{IF}/\texttt{EVALUATE} cases &
Control-flow branching across operations (e.g., add, delete, inquire) &
Path-sensitive (if feasible) or flow-sensitive fallback \\
\hline
Requests/responses around program calls &
Requests and responses used before and after program calls &
Without call-chain (WoCC) \\
\hline
Multiple calls to the same program &
Context sensitivity and overlapping parameter usage &
With call-chain (WCC) \\
\hline
Chained assignments &
Transitive data flow across variables &
Flow-sensitive \\
\hline
Deeply nested records or \texttt{REDEFINES} usage &
Shared memory and aliasing among overlapping records &
Flow-sensitive with heuristic pruning \\
\hline
\end{tabular}
\caption{Common COBOL Code Patterns and Recommended Analysis Variants}
\label{tab:patterns2}
\end{figure*}

\begin{figure}
    \includegraphics[width=89mm]{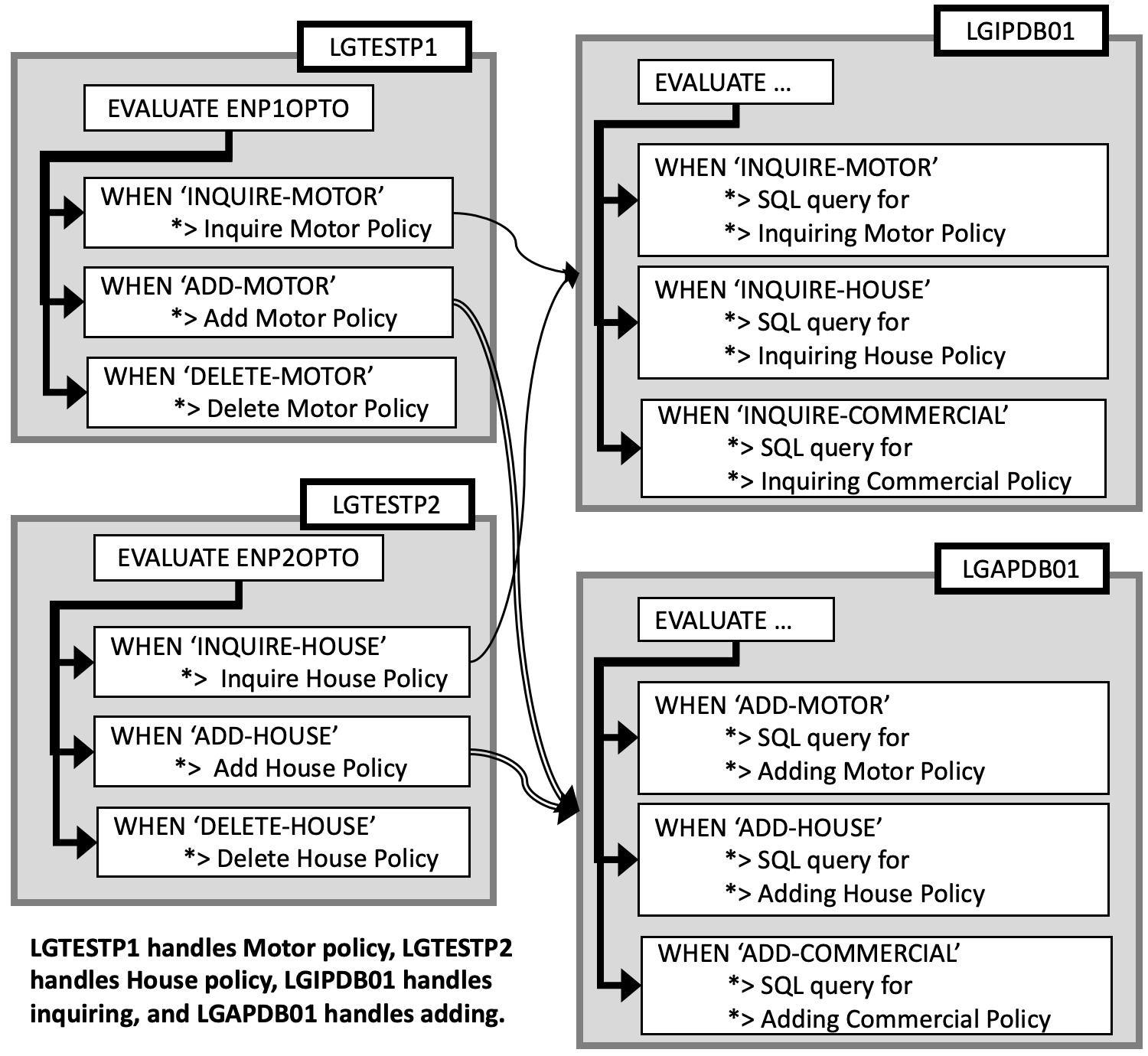} 
    \caption{Sequential {\ttfamily EVALUATE} (Switch) or {\ttfamily IF}-{\ttfamily ELSE} Cases}
    \label{fig:pattern1}
\end{figure}
\begin{figure}
    \includegraphics[width=89mm]{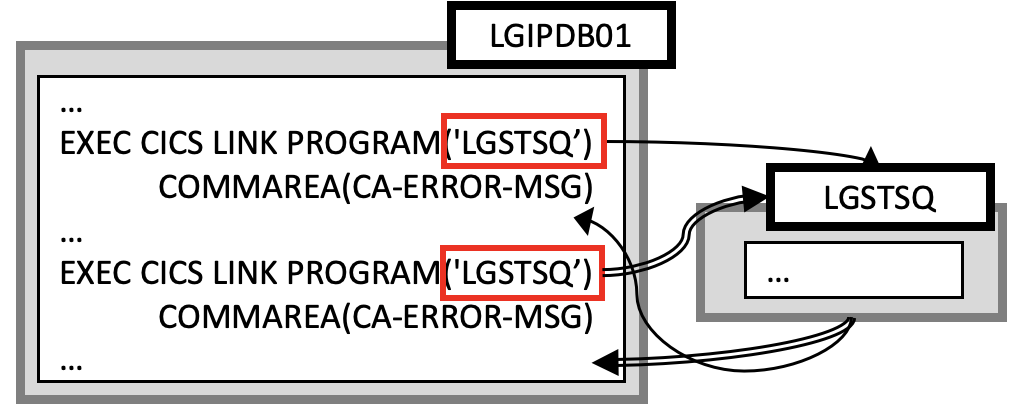} 
    \caption{Multiple Calls to the Same Program}
    \label{fig:pattern3}
\end{figure}
\begin{figure}
    \includegraphics[width=89mm]{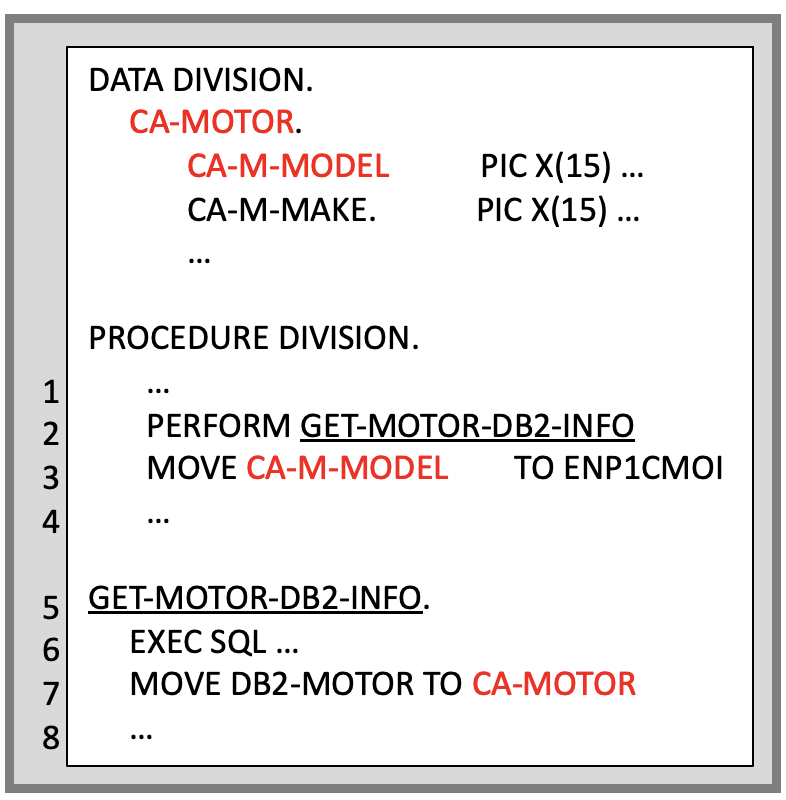} 
    \caption{Chained or Transitive Assignment}
    \label{fig:pattern4}
\end{figure}
\begin{figure}
    \includegraphics[width=89mm]{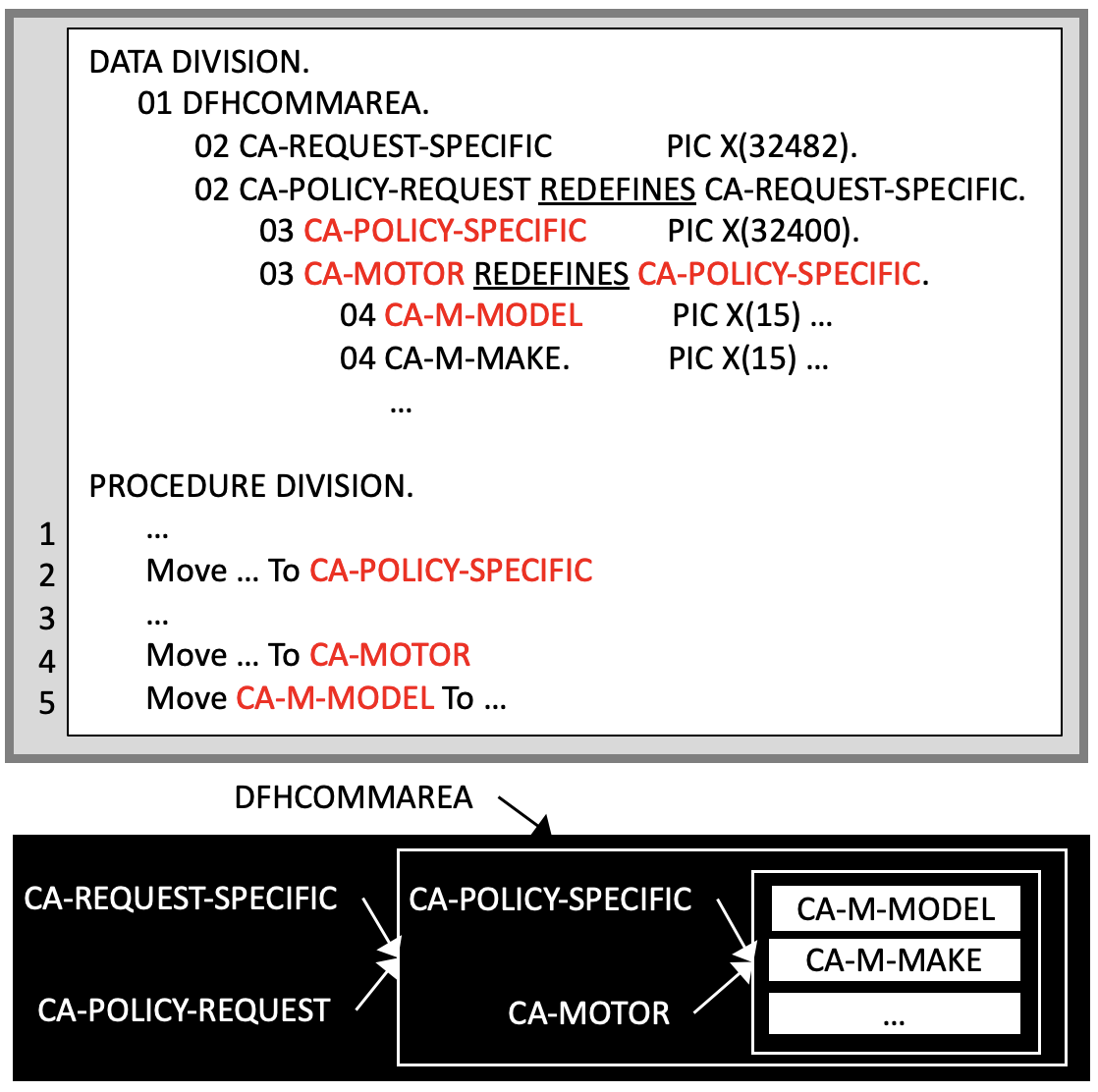} 
    \caption{Deeply Nested Records and Shared Memory via REDEFINES}
    \label{fig:pattern5}
\end{figure}

These guidelines enable practitioners to select configurations matching code idioms, improving precision and efficiency. For example, GENAPP predominantly uses P2, favoring WoCC.

\begin{itemize}
\item {Sequential {\ttfamily EVALUATE} (Switch) or {\ttfamily IF}-{\ttfamily ELSE} Cases}: Figure~\ref{fig:pattern1} shows extensive use of switch statements to handle house, motor, and other policies for inquiry, addition, and deletion operations. This common COBOL coding style highlights the need to use {\em path-sensitive analysis}. However, since path-sensitive analysis is time consuming, we suggest the user to extract code blocks across {\ttfamily EVALUATE} (switch) cases or {\ttfamily IF}-{\ttfamily ELSE} cases to form a unified a linear refactored program. This would allow a {\em flow-sensitive} or {\em flow-insensitive} to produce precise requests/responses.
\item {Storing Requests and Responses Before and After a Program Call}: Figure~\ref{fig:input} demonstrates this pattern, where request and response fields can be computed from the code preceding and succeeding a program call, respectively, {\em without performing call chain analysis}.
\item {Multiple Calls to the Same Program}: Figure~\ref{fig:pattern3} illustrates the necessity of {\em with call chain}, context-sensitive analysis to distinguish between calls made to the same program from different statements.
\item {Chained or Transitive Assignment}: Figure~\ref{fig:pattern4} indicates that accurately computing the request for this pattern requires {\em flow-sensitive analysis}. For example, although {\ttfamily CA-M-MODEL} is read in statement 3, it is not considered a request variable because its value is transitively populated from {\ttfamily DB2-MOTOR} in statement 7.
\item {Deeply Nested Records and Shared Memory via REDEFINES}: Figure~\ref{fig:pattern5} depicts the complex coding practices in COBOL, including the interchangeable use of record names and member names. This pattern underscores the need for automation, as manual computation of requests and responses in such scenarios would be error-prone.
\end{itemize}

\section{Integrating Extracted Code into a Compile-able Program}
\label{sec:integrate}

We create a single compile-able program that integrates code blocks from one or more programs to achieve a single business functionality, which we refer to as a {\em refactored program}. This refactored program can both be called by and call other programs. Below is our approach to refactoring code blocks to create a standalone compile-able program.

\subsection{Declarations in Data Division}

In COBOL, the Working-Storage Section is used to declare local variables that persist throughout the execution of the program, while the Linkage Section is used to declare variables passed from the calling program. 

{\bf Working-Storage Section}: 
We declare only those variables that are used in the code blocks, rather than all the variables. The variables used can be from programs or copybooks. To maintain the record structure of the variables, we iteratively retrieve their parents until the ancestor level and their children until the leaf level. If any variable has a {REDEFINES} (shared memory) or an {OCCURS} (array) depending on clause, we retrieve those variables as well, including their parents and children. This avoids potential compilation issues due to missing variables.

For example, in Figure~\ref{fig:pattern5}, since {\ttfamily CA-MOTOR} is used, we need to declare not only {\ttfamily CA-MOTOR} but also its children ({\ttfamily CA-M-MODEL}, {\ttfamily CA-M-MAKE}), its parents ({\ttfamily CA-POLICY-REQUEST}, {\ttfamily DFHCOMMAREA}), and redefined variables ({\ttfamily CA-POLICY-SPECIFIC}, {\ttfamily CA-REQUEST-SPECIFIC}) in the data division.

{\bf File Section:} If the code blocks interact with files, the relevant file variables are retrieved and placed in the file section and the file-control section of the input-output section.

{\bf Linkage Section:}
A call graph is created with all the programs to identify the starting program, which is the entry point to the refactored program with no incoming program-program edge. If the starting program contains a linkage section, the refactored program will contain linkage variables from it. The intersection of the used variables and the linkage section of the original starting program becomes the linkage section of the refactored program. For example, in Figure~\ref{fig:ex2}, when code blocks are selected from {\ttfamily LGTESTP1}, {\ttfamily LGIPOL01}, and {\ttfamily LGIPDB01}, the linkage section of the refactored program will contain variables from {\ttfamily LGTESTP1} since it is the starting program. Non-starting program variables are placed in the working-storage section, as these programs are no longer external.

{\bf Cursors:}
A cursor in COBOL is a database concept used to handle multiple rows retrieved from a database query. If the code blocks contain statements that open, fetch, or close a cursor, the cursor declaration is retrieved and populated as part of the working-storage section.

{\bf DCLGEN:}
DCLGEN (Declaration Generator) is a tool in DB2 that automatically generates COBOL and SQL declarations for database tables. If code blocks contain statements that interact with DB2 tables, and if the declaration is present in DCLGEN copybooks, the DCLGEN copybook statement is populated in the working-storage section.

\subsection{De-duplicating Variable and Paragraph/Section Names}

Duplicate variable and section/paragraph names are common when combining code blocks from multiple programs.

{\bf Variable Names:} Duplicate ancestor variables (01 level) are renamed, as qualifiers won’t resolve them. In programs, duplicates can be renamed or qualified, while in copybooks, adding qualifiers suffices. Renaming involves appending unique suffixes (-1, -2, etc.) to all occurrences. To optimize, structurally identical duplicates are eliminated, reducing redundancy and size of the refactored program.

{\bf Paragraph and Section Names:} Duplicate section or paragraph names may also occur. To resolve this, we rename each occurrence of the duplicate paragraph/section name to a unique name by adding suffixes (-1, -2, etc.), ensuring each suffix is unique to the program.

{\bf DFHCOMMAREA:} {\ttfamily DFHCOMMAREA} is a unique variable; there can only be one instance of it in a program, and it must be declared in the linkage section. Duplicates are renamed to avoid conflicts. For example, in Figure~\ref{fig:ex3}, {\ttfamily DFHCOMMAREA} from {\ttfamily LGTESTP1}, {\ttfamily LGIPOL01}, and {\ttfamily LGIPDB01} have been de-duplicated by suffixing -1, -2, and -3, respectively in the code blocks.

\subsection{Inter-program Calls}

When code blocks from both the calling and called programs are present, inter-program calls like {\ttfamily CALL}, {\ttfamily XCTL}, or {\ttfamily EXEC CICS LINK} are replaced with intra-program calls using a {\ttfamily PERFORM} statement since both code blocks are now in the same refactored program. For each program call, the user needs to specify the first-paragraph and last-paragraph to be called in place of the program call. If the last-paragraph is not specified, we create a new paragraph immediately after the first-paragraph, with the same name but with an "-exit" suffix, ensuring it contains an {\ttfamily EXIT} statement. The program call in the refactored program is replaced with a {\ttfamily PERFORM} {$\langle$first-paragraph$\rangle$} {\ttfamily THRU} {$\langle$last-paragraph$\rangle$} statement. We then ensure that information is passed correctly between the calling and called paragraphs by adding {\ttfamily MOVE} statements before and after the paragraph call. This ensures seamless program stitching, preventing erroneous data flow between the programs. For example, Figure~\ref{fig:ex3} shows this in calls to paragraphs of programs {\ttfamily LGIPOL01} and {\ttfamily LGIPDB01}.

\subsection{Formatting}
A byproduct of renaming and adding suffixes to handle duplicates may be exceeding the 80-character line limit in COBOL. To mitigate this, long lines are split to preserve the system-generated number area (columns 73-80), ensuring the total line length doesn’t exceed 80 characters. The new line must have the same indicator area (column 7) as the original line. For example, if a line is a comment, the new line must also be a comment. If the original line had quotes, the new line must start with quotes.

\subsection{Incomplete Code Blocks}
These can lead to compilation issues. To address this:  
\begin{itemize}
    \item If a code block lacks a closing period and the next block starts a new paragraph, a period is added to the last statement to ensure proper termination.  
    \item If the block begins with {\ttfamily WHEN}, {\ttfamily ELSE}, or {\ttfamily ELSE IF}, i.e. without the corresponding {\ttfamily SWITCH} or {\ttfamily IF}, then the {\ttfamily WHEN} and {\ttfamily ELSE} statements are commented out. For {\ttfamily ELSE IF}, the {\ttfamily ELSE} is removed while keeping the {\ttfamily IF} condition intact to preserve logical correctness.
\end{itemize}

\subsection{IMS TM Interactions}
Code blocks interacting with IMS TM (transaction management system) must handle Program Communication Blocks (PCBs) correctly. Only the PCBs referenced in the code blocks are included in the {\ttfamily USING} clause in the Procedure Division.   The IOPCB is always required and placed first. If the IMS Entry statement is missing, it is fetched and inserted at the start of the Procedure Division. 
Just like in structured program analysis, where missing dependencies must be reconstructed to maintain correctness, handling IMS interactions requires careful retrieval and placement of essential elements.

In conclusion, our automation transforms user-selected code blocks from the monolith application into a standalone, compile-able program.
\section{Code Refactor}
\label{sec:refactor}
To enable communication, both the exposed code (Section~\ref{sec:exposed}) and the caller code (Section~\ref{sec:caller}) need to be refactored.

\subsection{Exposed Code}
\label{sec:exposed}
Starting from the artifacts/seeds like transactions, screens, business rules, etc., we extract code blocks that need to be exposed as APIs by slicing the code. However, these slices may include processes that are not relevant to the API. To address this, we propose further refactoring of the sliced code in the following cases:

\begin{itemize}
\item Terminal commands that appear in the API should be removed or placed under an if-else condition to prevent their execution as part of the API but allow them to be executed otherwise.

\item DB table SQL queries need to be optimized to fetch only the required fields for the API. Similarly, any other external resources, such as VSAM files or MQ, should be accessed only if required by the service.

\item Sanity checks that might have existed in the code need to be removed from the API or placed under an if-else condition to avoid their execution as part of the API but still execute them when needed.

\end{itemize}

After refactoring, the code can be exposed using any connect tool, like IBM zOS Connect~\cite{zosc}. This tool generates a definition file for the API and includes schema definitions for the determined requests and responses in the form of Swagger.

\subsection{Caller Code}
\label{sec:caller}
Once the code is exposed as APIs, the caller code needs to be refactored as follows:
\begin{inparaenum}
\item For callers inside the application, they should now call the exposed code directly.
\item For callers outside the application, they should call the APIs using REST API calls.
\end{inparaenum}
Refactoring is done without duplicating the exposed code to ensure the maintenance of the applications.

For callers that need to make REST API calls, they utilize communication stubs of IBM zOS Connect. The communication stubs internally call the exposed code blocks from zOS Connect. During this process, the arguments of the caller are passed to the parameters or the request fields of the APIfied code block. Data structure mapping of parameters and arguments is accomplished using byte alignment~\cite{byte-alignment}. This ensures seamless integration and communication between the caller code and the exposed APIs.
\section{Making the APIs Production Ready}
\label{sec:production}

Even after identifying the necessary APIs to be introduced, there is lot of manual effort involved in making the APIs production ready. 
Below are the automation opportunities in the APIfication task:
\begin{itemize}
%\item Accelerate interface discovery and design. Reduce the effort to identify optimal API insertion points in the legacy code. Create request response design and corresponding copybooks.
\item Ensure API readiness. Send proper content-type header, use proper HTTP status code, handle errors properly, send error messages for client errors, ease debugging efforts by introducing access and error logs, handle API security etc.
\item Handle data dependencies before and after API calls. In other words, detect and resolve any friction that might arise in the overall process behaviour because of the introduction of new APIs. For example, APIs may be preceded by actions to prepare the data in the dependent DB. Such dependencies can be alerted and refactoring suggestions can be given to the developer.
%\item Map domain API specifications. Map domain-driven API specifications and semantics to legacy code and identify domain specific coverage.
\item Test the refactored code. Developing valid test cases in the host application and pass customer UAT (user acceptance testing).
\item Deployment with CICD pipelines. Develop CICD pipelines in DevOps~\cite{devops}, deploy services and APIs, and register with tools like API Connect~\cite{apiconnect}.
\end{itemize}
\section{Formalization and Proofs}
\label{sec:formal}

\renewcommand\qedsymbol{} 
\newcommand{\req}{\mathrm{Req}}
\newcommand{\resp}{\mathrm{Resp}}
\newcommand{\inn}{\mathrm{In}}
\newcommand{\out}{\mathrm{Out}}
\newcommand{\ereq}{\widehat{\mathrm{Req}}}
\newcommand{\eresp}{\widehat{\mathrm{Resp}}}
\newcommand{\succn}{\mathrm{succ}}
\newcommand{\gen}{\mathrm{Gen}}
\newcommand{\killn}{\mathrm{Kill}}
\newcommand{\stmts}{\mathrm{stmts}}

In Section~\ref{sec:eq-static}, we formalize the computation of requests/responses at the intra-program level using two different static analyses:
\begin{itemize}
\item Flow-sensitive static analysis: This analysis computes requests/responses by considering the data flow dependencies within the code, taking into account the order of reads and writes to variables. 
\item Flow-insensitive static analysis: This analysis computes requests and responses without considering the order of reads and writes to variables, making it less precise but faster to compute.
\end{itemize}

In Section~\ref{sec:eq-exec}, we define the execution semantics of requests and responses to understand their behavior during program execution.

In Section~\ref{sec:eq-sound}, we prove that both static analyses are sound, meaning that they provide correct and valid results with respect to the execution semantics. Furthermore, we demonstrate that the flow-sensitive analysis is at least as precise as the flow-insensitive analysis, ensuring that the flow-sensitive analysis provides more accurate results.

\subsection{Key Terminologies of Static Analysis}
In static analysis, data flow equations are mathematical representations used to model the flow of data (in our case request/response variables) within a code block without executing the code block (Section~\ref{sec:eq-static}). These equations define how data flows from one program point to another. In flow-sensitive analysis, we suffix data with $\inn_n$ and $\out_n$ to refer to the program point before and after statement $n$, respectively. In flow-insensitive analysis, we subscript data with $(s,e)$ to denote the code block from statement numbers $s$ to $e$. Terms $\gen_n$ and $\killn_n$ denote sets to model the effects of program statement $n$ on the computed data.

Execution semantics (Section~\ref{sec:eq-exec}) in static analysis  refers to understanding how a code block affects the data. These are also represented using data flow equations. Suffixes $\inn_{p,n}$ and $\out_{p,n}$ refer to program points before and after statement $n$ for an execution path $p$. Also, subscript $p, (s,e)$ denotes the code block from statement numbers $s$ to $e$ along execution path $p$.

\subsection{Computation using Static Analysis}
\label{sec:eq-static}
The data flow equations in this section compute request variables in  $\req\inn_n$ and $\req\out_n$ and response variables in $\resp_{(s,e)}$ using static analyses.
$\req\gen_n$ denotes the set of variables used in statement $n$. $\req\killn_n$ and $\resp\gen_n$ represent identical sets of variables defined in statement $n$, but they are utilized for defining requests and responses, respectively. Function $\succn(n)$ denotes the set of statements succeeding $n$ along all possible control flow paths. Function $\stmts(s,e)$ denotes the set of statements in the code block from statement $s$ to $e$.

\emph{Flow-sensitive static analysis}
\begin{equation}
\req\inn_n = \req\gen_n \cup (\req\out_n - \req\killn_n)
\end{equation}
\begin{equation}
\req\out_n = 
\begin{cases}
    \{ \} & n \emph{ is end statement} \\
     \cup_{k \in \succn(n)} \req\inn_k & \mathrm{otherwise} 
\end{cases}
\end{equation}
\begin{equation}
\resp_{(s,e)} = \cup_{k \in \stmts(s,e)} \resp\gen_k
\end{equation}

\emph{Flow-insensitive static analysis}
\begin{equation}
\req_{(s,e)} = \cup_{k \in \stmts(s,e)} \req\gen_k
\end{equation}
\begin{equation}
\resp_{(s,e)} = \cup_{k \in \stmts(s,e)} \resp\gen_k
\end{equation}

\subsection{Computation using Execution Semantics}
\label{sec:eq-exec}
The data flow equations in this section compute request variables in $\ereq\inn_{p,n}$ and $\ereq\out_{p,n}$ and response variables in $\eresp_{p,(s,e)}$.
Function $\succn(p,n)$ denotes the set of statements succeeding $n$ along control flow path $p$. Function $\stmts(p, s, e)$ denotes the set of statements in the code block from statement $s$ to $e$ along $p$.

\begin{equation}
\ereq\inn_{p,n} = \req\gen_n \cup (\ereq\out_{p,n} - \req\killn_n)
\end{equation}
\begin{equation}
\ereq\out_{p,n} = 
\begin{cases}
    \{ \} & n \emph{ is end statment} \\
     \ereq\inn_{p,\succn(p,n)} & \mathrm{otherwise} 
\end{cases}
\end{equation}
\begin{equation}
\eresp_{p,(s,e)} = \cup_{k \in \stmts(p,s,e)} \resp\gen_k
\end{equation}

\subsection{Soundness and Precision}
\label{sec:eq-sound}
In this section, we prove soundness and precision of computed request/response variables. The request/response variables for the code block $(s,e)$ are accumulated in the following: 
\begin{itemize}
    \item $\req\inn_s$ and $\resp_{(s,e)}$ in flow-sensitive analysis
    \item $\req_{(s,e)}$ and $\resp_{(s,e)}$ in flow-insensitive analysis
    \item $\cup_{\forall p} \ereq\out_{p,e}$ and $\cup_{\forall p} \eresp_{p,(s,e)}$ in execution semantics
\end{itemize}

\begin{theorem}
{\em Soundness}. Our proposed analyses compute sound values with respect to the execution semantics.
\end{theorem}
%\vspace{-0.3cm}
\begin{proof}
From the data flow equations, it can be seen that the set of request and response variables computed by static analyses (flow-sensitive and flow-insensitive) is a superset of those computed by execution semantics for all possible execution paths. This is because, generation is over-approximated and killing is under-approximated in the proposed analyses with respect to the execution semantics.
\begin{equation*}
\req\inn_s \supseteq \cup_{\forall p} \ereq\inn_{p,s}
{\ \emph {and} \ }
\req_{(s,e)} \supseteq \cup_{\forall p} \ereq\inn_{p,s}
\end{equation*}
\begin{equation*}
\resp_{(s,e)} \supseteq \cup_{\forall p} \eresp_{p,(s,e)}
\end{equation*}
\end{proof}
%\vspace{-0.5cm}
\begin{theorem}
{\em Precision}. Flow-sensitive analysis is at least as precise as the flow-insensitive analysis.
\end{theorem}
%\vspace{-0.3cm}
\begin{proof}
It can be seen that the flow-sensitive computation of request variables forms a subset of the flow-insensitive computation of request variables. However, the analysis yields equivalent response variables in both cases.
\begin{equation*}
\req\inn_s \subseteq \req_{(s,e)}
\end{equation*}
\end{proof}
%\vspace{-0.7cm}

By formalizing and proving the correctness of these analyses, we establish a solid foundation for the automated identification of requests and responses, allowing us to accurately extract APIs from the mainframe applications.
\section{Experiments and Qualitative Study}
\label{sec:exp}

The implementation of our API signature computation is part of IBM Watsonx Code Assistant for Z Refactoring Assistant~\cite{ibmra}. We have implemented both with and without call chain analyses, in conjunction with both flow-insensitive and flow-sensitive static analyses. We have not implemented the path-sensitive approach because in practice, analyzing all possible control flow paths for large, complex software can be extremely time-consuming and resource-intensive, if not impossible.

For our experiments, we utilized the GENAPP application~\cite{genapp}, a public industry application namely CBSA~\cite{cbsa}, and a private industry application, with details anonymized to protect confidentiality. Conducting experiments on mainframe systems and implementing COBOL parsing in Python presented challenges, as the availability of COBOL/mainframe developers is limited, despite the widespread reliance on mainframes in the industry~\cite{mainframe}.

GENAPP comprises 36 programs, 5 copybooks, and 8 DB2 tables, involving a total of 7,566 variables. From this application, our system identified several candidate APIs using the approaches discussed in Section \ref{sec:id}. For our experiments, we selected the following APIs:
\begin{itemize}
\item Transaction-based APIs: 
Inquire motor policy (TI),
Add motor policy (TA),
and Update motor policy (TU).
\item Data access-based APIs:
Inquire motor policy (DI),
Add motor policy (DA),
and Update motor policy (DU).
\end{itemize}

The private industry application, which deals with Bill Proration for a telecommunication company consists of 34 programs, 134 copybooks, and 123 DB2 tables, with a total of 11,543 variables. We referred to their business requirements document and identified the following candidate APIs:
\begin{itemize}
\item API1: Check for a suitable prorated billing plan and associated pricing for the given advance period and apply it. Contains business rules and functionality.
\item API2: Accumulate the prorated billing amounts into buckets. Contains business rules and functionality.
\item API3: Check billing plan that matches the next cycle start date and has not been backdated, indicating its availability for advance application. Contains data accessing points.
\item API4: Financial market. Contains business functionality.
\item API5: Handle booking. Contains business functionality.
\end{itemize}

CBSA~\cite{cbsa} is a CICS Bank Sample Application, which simulates the operation of a bank from the point of view of the Bank Teller. It contains 46 programs, 9 CICS transactions, 3 DB2 tables, 31 copybooks. We referred to their RESTful API guide and identified the following APIs:
\begin{itemize}
\item API1: Gets data from account table. Contains data accessing points.
%INQACC-fetchdata para only
\item API2: Updates account details. Contains business rules, functionality, and data accessing points.
%DBCRFUN-UPDATE-ACCOUNT-DB2 and its slice
\item API3: Read customer file. Contains data accessing points.
%INQCUST-READ-CUSTOMER-VSAM para only
\item API4: Utility program to write errors in a centralized datastore. Contains a data accessing point.
%ABNDPROC
\item API5: Write a record to the customer VSAM file. Contains business rules, functionality, and data accessing points.
%CRECUST-WRITE-CUSTOMER-VSAM
\end{itemize}

With these selected APIs, we conducted our experiments to validate and assess the effectiveness of our proposed methodology.

\subsection{Comparison of Manual vs. Automated API Identification}
Below is a comparison of manual versus our proposed automated steps in the APIfication of GENAPP.

\begin{itemize}
\item \textbf{Manual API Identification}: We engaged five Subject Matter Experts (SMEs) with an average experience of 10 years, having expertise in mainframe modernization, APIfication, and a good understanding of the GENAPP code. The SMEs manually listed 15 transaction-based APIs and 10 data accessing APIs in the GENAPP application.
\item \textbf{Automatic API Identification}: Leveraging our automated methodology,  our system identified 15 transaction-based APIs, 21 procedure-based APIs that do not involve data access, and 21 data accessing APIs.
\end{itemize}

As observed from the comparison, we successfully identified the same transaction-based APIs as the manual process, ensuring no missed APIs. Our methodology exhibits a coverage of 100\%, meaning that all APIs in GENAPP were accurately identified without any omissions.

\subsection{API Identification}
We engaged five Subject Matter Experts (SMEs) with an average experience of 10 years, having expertise in mainframe modernization, APIfication, and a good understanding of the GENAPP code. The SMEs manually listed 15 transaction-based APIs and 10 data accessing APIs in GENAPP.
Leveraging our list of artifacts,  our system identified 15 transaction-based APIs, 21 procedure-based APIs that do not involve data access, and 21 data accessing APIs.

As observed from the comparison, we successfully identified the same transaction-based APIs as the subject matter experts, ensuring no missed APIs. Our methodology exhibits a coverage of 100\%, meaning that all APIs in GENAPP were accurately identified.

\begin{figure}
    \centering
    \begin{tabular}{|l|l|l|l|l|}
    \hline
    \multirow{2}{*}{API} 
    & \multicolumn{2}{c|}{Without call chain}
    & \multicolumn{2}{c|}{With call chain} \\ \cline{2-5}

    & FI & FS & FI & FS
    \\ \hline \hline

    TI & 17.57 & 25.42 & 49.48 & 147.94 \\ \hline
    TA & 12.51 & 23.65 & 120.72 & 453.33  \\ \hline
    TU & 23.19 & 48.59 & 168.79 & 540.40 \\ \hline
    DI & 14.07 & 26.07 & 13.58 & 29.13 \\ \hline
    DA & 18.73 & 27.99 & 14.85 & 31.42 \\ \hline
    DU & 10.61 & 23.74 & 11.10 & 21.27 \\ \hline
    \end{tabular}
    \Description{}
    \caption{Analysis time in seconds of the four variants of computing requests and responses. These are GENAPP APIs explained in Section~\ref{sec:exp}. FI and FS stand for Flow-Insensitive and Flow-Sensitive.}
    \label{fig:exp1}
\end{figure}

\begin{figure}
    \small 
    \setlength\tabcolsep{3.5 pt}
    \centering
    \begin{tabular}{|l|l|l|l|l|l|l|l|l|}
    \hline
    \multirow{3}{*}{API} 
    & \multicolumn{4}{c|}{Without call chain}
    & \multicolumn{4}{c|}{With call chain} \\ \cline{2-9}

    & \multicolumn{2}{c|}{FI}
    & \multicolumn{2}{c|}{FS} 
    & \multicolumn{2}{c|}{FI}
    & \multicolumn{2}{c|}{FS} 
    \\ \cline{2-9}
    
    & Req & Res
    & Req & Res
    & Req & Res
    & Req & Res \\ \hline \hline

    TI & 0+2 & 0+11 & 0+2 & 0+11 & 0+45 & 0+73 & 0+2 & 0+73 \\ \hline
    TA & 0+12 & 0+4 & 0+12 & 0+4 & 0+56 & 0+48 & 0+34 & 0+48 \\ \hline
    TU & 0+13 & 0+14 & 0+11 & 0+14 & 0+90 & 17+91 & 0+4 & 17+91 \\ \hline
    DI & 0+29 & 17+23 & 0+20 & 17+23 & 0+29 & 17+23 & 0+20 & 17+23 \\ \hline
    DA & 9+5 & 0+5 & 5+0 & 0+5 & 9+5 & 0+5 & 5+0 & 0+5 \\ \hline
    DU & 9+5 & 0+5 & 5+0 & 0+5 & 9+5 & 0+5 & 5+0 & 0+5 \\ \hline
    \end{tabular}
    \Description{}
    \caption{Number of working storage + linkage section fields for requests and responses using the four variants of computation. These are GENAPP APIs explained in Section~\ref{sec:exp}. FI and FS stand for Flow-Insensitive and Flow-Sensitive. Req and Res stand for Requests and Responses.}
    \label{fig:exp2}
\end{figure}

\begin{figure}
    \small 
    \setlength\tabcolsep{3.5 pt}
    \centering
    \begin{tabular}{|l|l|l|l|l|l|l|l|l|}
    \hline
    \multirow{3}{*}{API} 
    & \multicolumn{4}{c|}{Without call chain}
    & \multicolumn{4}{c|}{With call chain} \\ \cline{2-9}

    & \multicolumn{2}{c|}{FI}
    & \multicolumn{2}{c|}{FS} 
    & \multicolumn{2}{c|}{FI}
    & \multicolumn{2}{c|}{FS} 
    \\ \cline{2-9}
    
    & Req & Res
    & Req & Res
    & Req & Res
    & Req & Res \\ \hline \hline

    API1 & 9 & 4 & 9 & 4 & 9 & 4 & 9 & 4 \\ \hline
    API2 & 5 & 4 & 5 & 4 & 5 & 4 & 5 & 4 \\ \hline
    API3 & 19 & 15 & 19 & 15 & 19 & 15 & 19 & 15  \\ \hline
    API4 & 31 & 21 & 31 & 21 & 221 & 177 & 31 & 177 \\ \hline
    API5 & 21 & 20 & 21 & 20 & 211 & 176 & 21 & 176 \\ \hline
    \end{tabular}
    \Description{}
    \caption{Number of fields for requests and responses using the four variants of computation. These are private industry application APIs explained in Section~\ref{sec:exp}. FI and FS stand for Flow-Insensitive and Flow-Sensitive. Req and Res stand for Requests and Responses.}
    \label{fig:exp4}
\end{figure}

\begin{figure}
    \small 
    \setlength\tabcolsep{3.5 pt}
    \centering
    \begin{tabular}{|l|l|l|l|l|l|l|l|l|}
    \hline
    \multirow{3}{*}{API} 
    & \multicolumn{4}{c|}{Without call chain}
    & \multicolumn{4}{c|}{With call chain} \\ \cline{2-9}

    & \multicolumn{2}{c|}{FI}
    & \multicolumn{2}{c|}{FS} 
    & \multicolumn{2}{c|}{FI}
    & \multicolumn{2}{c|}{FS} 
    \\ \cline{2-9}
    
    & Req & Res
    & Req & Res
    & Req & Res
    & Req & Res \\ \hline \hline

    API1 & 11 & 19 & 11 & 19 & 24 & 47 & 24 & 47 \\ \hline
    API2 & 38 & 22 & 38 & 22 & 39 & 47 & 38 & 47 \\ \hline
    API3 & 9 & 0 & 9 & 0 & 10 & 21 & 10 & 21 \\ \hline
    API4 & 4 & 3 & 4 & 3 & 4 & 3 & 4 & 3\\ \hline
    API5 & 39 & 30 & 39 & 30 & 39 & 54 & 39 & 54 \\ \hline
    \end{tabular}
    \Description{}
    \caption{Number of fields for requests and responses using the four variants of computation. These are CBSA APIs explained in Section~\ref{sec:exp}. FI and FS stand for Flow-Insensitive and Flow-Sensitive. Req and Res stand for Requests and Responses.}
    \label{fig:exp5}
\end{figure}

\subsection{Static Analysis Time}
Figure~\ref{fig:exp1} presents a list of six identified APIs in GENAPP, along with their analysis timings conducted on a personal computer with the following configuration: 2.3 GHz Quad-Core Intel Core i7, macOS Ventura, 32GB RAM. The figure displays the analysis time taken to find requests and responses for each API using four different variants of static analysis, encompassing all combinations of with/without call chain analyses and flow-insensitive/sensitive analyses. The timings include fetching control flow graphs and program call information from a remote machine hosting ADDI (COBOL parser)~\cite{addi}, with the actual static analysis performed on the local computer. Key comparisons are provided below:

\begin{itemize}
\item \textbf{Flow-sensitive vs. Flow-insensitive}: Flow-insensitive analysis exhibits a shorter analysis time compared to flow-sensitive analysis, as the former is computationally simpler.
\item \textbf{With and Without Call Chain Analysis}: The analysis time for without call chain analysis is lower than with call chain analysis since the former does not involve traversing any called programs. Notably, the analysis time for APIs DA, DI, and DU, which do not have any calls, remains similar between the with and without call chain analysis variants. Overall, our recommended variant of without call chain analysis yields the most efficient results.
\end{itemize}

\subsection{Requests and Responses}
Figures~\ref{fig:exp2}, \ref{fig:exp4}, and \ref{fig:exp5} present the count of request and response fields for GENAPP, the private industry application, and CBSA respectively. These counts are computed using the four variants of static analysis, including all combinations of with/without call chain analyses and flow-insensitive/sensitive analyses. In Figure~\ref{fig:exp2}, we show the number of working storage and linkage section fields used as requests and responses for each API in GENAPP. The linkage section contains the variables used to send/receive data to/from the called program, while the working storage contains variables used by the program. Key comparisons of the variants used for computing requests and responses are as follows:

\begin{itemize}
\item \textbf{Flow-sensitive vs. Flow-insensitive}: Flow-sensitive analysis is the most precise, correctly identifying some fields that should not be marked as request since they are first written and then read. Notably, in analyzing the private industry application's API4, flow-insensitive analysis identifies 221 fields, while flow-sensitive analysis identifies only 31 fields. It is observed that the additional 190 fields identified by flow-insensitive analysis are spurious and unnecessary as request fields.
\item \textbf{With and Without Call Chain Analysis}: Computation using without call chain analysis may miss some request and response fields present in the called program. However, GENAPP and the industry applications are coded in a way that all requests and responses are present before and after the program call, not inside it. Therefore, without call chain analysis yields precise results. With call chain analysis, on the other hand, adds spurious fields when analyzing the called programs. However, code blocks for DA, DI, DU APIs in GENAPP do not have any calls, and similarly, code blocks for API1, API2, and API3 in the private industry application do not have any calls. Therefore, computation using with and without call chain analysis produces the same fields. Overall, our recommended variant of without call chain analysis produces the most precise results for GENAPP.
\end{itemize}

\subsection{Demonstrating Need of Automatic APIfication}
Automating the process of identifying requests and responses in mainframe applications is essential due to the complexity and size of the data-stores involved. Mainframe applications typically define large data-stores that are exchanged between programs, making the manual identification of relevant requests and responses a tedious and error-prone task.

To highlight the difference in the number of fields in the data-stores compared to the requests and responses for each API, we conducted a comparison (Figure~\ref{fig:exp2}). In GENAPP, there are two copybooks, namely LGCMAREA and SSMAP, with a total of 593 fields. In the original monolith, all these fields are passed in program calls. However, our static analysis results for six APIs in Figure~\ref{fig:exp2} demonstrate that only a fraction of these fields are used as requests and responses. The minimum number of linkage section fields used as requests/responses by an API is 2, and the maximum is 91. Similarly, in the private industry application, which contains 11,543 variables, only a few of them are identified as requests and responses.

This comparison clearly illustrates that manually identifying requests and responses from large and complex copybooks is challenging, time-consuming, and prone to errors. Automating the process using static analysis techniques significantly reduces the effort and ensures more accurate results. Therefore, automation is necessary to efficiently and reliably identify the relevant requests and responses for each API in mainframe applications.

Automating APIfication offers numerous advantages:

\begin{enumerate}
\item Eliminating the need for in-depth application knowledge: Automation allows developers to extract APIs without requiring extensive knowledge of the entire application. This reduces the reliance on specific experts and opens the process to a broader range of developers.

\item Flexibility in connector tool selection: Automated APIfication is not tied to a specific connector tool. This flexibility enables organizations to switch between different connector tools based on their requirements or changing technology landscape.

\item Easy traceability and continuous improvement: With automation, the APIfication steps are well-documented and easily traceable. This makes it simpler to analyze the process and identify areas for improvement. Continuous refinements can be made to enhance the overall efficiency.

\item Reduced delivery time and predictability: Automation streamlines the APIfication process, leading to quicker delivery times for APIs. Additionally, the automated approach makes it easier to estimate the time required for APIfication, providing more predictability in project planning.

\item Enhanced developer focus: By automating the monotonous and repetitive aspects of APIfication, developers can redirect their efforts towards adding new functionalities or composing innovative applications. This boosts productivity and encourages creativity in development work.
\end{enumerate}

Overall, automation empowers organizations with a more efficient and scalable approach to APIfication, resulting in increased productivity and improved application development processes.

\begin{figure}
    \centering
    \begin{tabular}{|p{0.38\textwidth}|p{0.04\textwidth}|}
    \hline
    API & Time \\ \hline \hline
    Inquire motor policy of a customer's claim (TI) & 0.58 \\ \hline
    Add motor policy of a customer (TA) & 0.60 \\ \hline
    Update motor policy of a customer (TU) & 0.49 \\ \hline
    Inquire commercial policy of a customer's claim  & 0.47 \\ \hline
    Inquire house policy of a customer's claim  & 0.61 \\ \hline
    Inquire endowment policy of a customer's claim  & 0.59 \\ \hline
    Inquire customer's claims for a policy & 0.50 \\ \hline
    Inquire customer's details for a claim & 0.62 \\ \hline
    Inquire policy numbers of a customer & 0.60 \\ \hline
    \end{tabular}
    \Description{}
    \caption{Execution time (seconds) when GENAPP APIs exposed on an IBM Z mainframe are called from the client side.}
    \label{fig:exp3}
\end{figure}

\subsection{Validation: LLM, Qualitative, and Execution}
\label{sec:exec}
Because of the \emph{lack of any ground‑truth}, our evaluation focuses on three complementary angles: LLM-as-a-judge, qualitative study, and an execution based manual check.

\textbf{LLM-as-a-judge.}
We used {\ttfamily gpt-oss-120b} (117 billion parameters;
131k tokens) as an LLM-as-a-judge. For each of the five APIs (API1–API5), we explicitly specified the lines to be APIfied and supplied only the COBOL programs that were directly or indirectly invoked by those lines. The model’s inferred request/response fields for all five APIs aligned with the flow-sensitive (FS) call-chain behavior of GENAPP. 

\textbf{Qualitative.}
Mainframe code ships vast data stores--593 fields in GENAPP and 11,543 in our proprietary app—yet Figure~\ref{fig:exp2} shows that each API needs only a small slice. 
We conducted a qualitative study involving nine senior engineers ($\approx$ 15 years' experience). These practitioners have a strong background in Mainframe application development and APIfication. They validated the extracted code, the request-response, and the generated standalone modules in GENAPP by our system. They estimated that finding those slices by hand would take about one week for GENAPP and several weeks for the larger app. With our tool they finished in a few minutes, and the speed‑up grew with application size.

\textbf{Execution.}
We exposed GENAPP APIs via z/OS Connect~\cite{zosc} and validated functional behavior on a Z system (GP-6, zIIP-2, 16\,GB RAM, z/OS v2.4, z14).

For validation, no publicly available ground truth exists for "correct’’ request/response sets in COBOL programs, so we cannot report supervised precision/recall. Instead, we rely on (i) LLM-as-a-judge, (ii) participant judgments / qualitative study and (iii) an end‑to‑end flow test described next.

\subsection{Qualitative Study}
\label{sec:study}
% Charles, Shami 30
% Shivali 16, Saravanan 20, Alex 3, Ashwin 5 years
% Richard Cadapeaud 23 years
% Divya Saxena 17 years
% Amit Kumar Singh 17 years
We conducted a qualitative study involving nine practitioners with an average work experience of approximately fifteen years. These practitioners have a strong background in Mainframe application development and APIfication, and they possess extensive knowledge of the challenges associated with these processes. Some of them also have direct client-facing experience in various industries, such as banking, retail, travel, and insurance.

During the study, we provided the practitioners access to GENAPP and the private industry application. We then asked them specific questions related to the manual identification of API endpoints and the computation of their signatures. Specifically, we inquired about any deterrents they encountered and time required to complete these tasks manually. Additionally, we requested them to prioritize the artifacts/seeds mentioned in Section~\ref{sec:id}.

The participants reported facing the same deterrents mentioned in Section~\ref{sec:intro}. They mentioned that manually identifying API endpoints and computing their signatures for the GENAPP application would take, on average, at least one week. For the more complex private industry application, they estimated it would take several weeks to complete these tasks manually.

Furthermore, the participants shared their priorities for API endpoint identification, stating that they would give higher importance to endpoints driven from screens, transactions/jobs, business rules, and data access points compared to other artifacts/seeds. 

When presented with the automated approach proposed in our paper, the practitioners acknowledged that these tasks would take only a few minutes. They were impressed with the efficiency and time-saving capabilities of the automated approach. As the size and complexity of the application increases, the practitioners believed that the time saved using our automated approach would become even more significant.

\subsection{APIfied Application's Execution Time}
\label{sec:exec}
Because of the \emph{lack of any ground‑truth}, our evaluation focuses on two complementary angles: the qualitative study and an execution based manual check described here.

To validate the APIfied GENAPP, we exposed TI, TA, TU APIs and several other APIs using zOS Connect \cite{zosc}. The APIs were hosted on an IBM Z machine with configurations: GP-6, zIIP-2, RAM-16G, zOS V2.4, z14 hardware. Then, these APIs were called and timed using the curl command. Figure~\ref{fig:exp3} shows the execution time in seconds of the APIs when these are called from the client side. We also manually validated the correctness of the API's functionality.

\subsection{LLMs versus Static Analysis}
We investigated whether large language models (LLMs) could independently infer request–response parameters from realistic COBOL programs without relying on static analysis. 
While a frontier-scale model such as {\ttfamily gpt-oss-120b} (117 billion parameters; 131k tokens) could act as a reasonably accurate “LLM-as-a-judge”, a slightly smaller model, {\ttfamily mistral 8×7B} (46.7 B parameters; 32k tokens), produced imprecise request–response fields and failed to identify the minimal set of parameters.
To evaluate mistral 8x7b, we instructed mistral 8x7b on the precise lines of code to be APIfied and supplied four COBOL programs—{\ttfamily LGTESTP1}, {\ttfamily LGIPOL01}, {\ttfamily LGIPDB01}, {\ttfamily LGSTSQ}—and one copybook—{\ttfamily LGCMAREA} (Figure \ref{fig:ex2}). 
Because LLMs lack inter-procedural understanding, they do not follow call chains or data-flow across modules; instead, they tend to over-approximate by treating most syntactically visible variables as request/response candidates. For example, it incorrectly marked {\ttfamily CA-LASTCHANGED}, {\ttfamily CA-BROKERID}, {\ttfamily CA-BROKERSREF}, and {\ttfamily CA-PAYMENT} as response fields, despite none of them being produced along the relevant execution path. These errors demonstrate that even when provided with only the relevant program text, LLMs—even strong mid-sized ones—hallucinate or misattribute request/response fields, reinforcing the need for semantics-aware static analysis to reliably extract correct API signatures. 

% =========================
\subsection{Threats to Validity}
% =========================
\textbf{Internal validity.}
Time savings are based on developer logs and self-reports; controlled studies could further isolate the effect of analysis vs.\ refactoring features.

\noindent\textbf{External validity.}
GENAPP, CBSA, and two enterprise systems cover batch/CICS/IMS styles, but generalization to all mainframe codebases remains to be confirmed.

\noindent\textbf{Construct validity.}
No public ground truth exists for ``correct'' request/response sets; we triangulated via execution tests and LLM-as-a-judge. Establishing a benchmark for API-signature inference is future work.
\section{Tool Evaluation with Client Applications}
\label{sec:client}

Here we present the results demonstrating how our tool improved developer productivity by exposing extracted code blocks as APIs for two client applications.  

\textbf{Client Case 1}: Our tool was tested with 10 business service use cases to gather requirements for isolating business services to be re-written in Java (from scratch, not through translation). One COBOL program relied heavily on copybooks containing procedural code, with bi-directional invocations between the program and copybooks. The client reported that our tool on code generation and request/response computation, was found to be useful for modularizing the code base into functional services based on the extracted feature set. It was particularly beneficial for transitioning to a modern stack or refactoring existing COBOL code within the mainframe environment. However, a limitation was that it required a functional Mainframe Subject Matter Expert (SME) to manually review, refine, and validate the output.  

\textbf{Client Case 2}: The tool was applied to five COBOL programs, each ranging from 600 to 4000 SLoC. The client reported a reduction in time required to write refactored code and reassemble business functions from code snippets by 15\%-50\%. A direct correlation was observed between the size of the program and the productivity benefits. The primary factors influencing these savings were the developer's application knowledge, the size of the original program, and the complexity of the program's flow graph and data usage.  The refactoring led to a reduction in cyclomatic complexity, reinforcing the value of modularization as a modernization pattern. This improvement enhanced the understanding and maintainability of the code, demonstrating the return on investment of using automation tools for modernization. 

\section{Related Works}
\label{sec:related}
As far as our knowledge extends, no other existing work offers a comprehensive list for identifying APIs and an automated approach for computing API signatures in mainframe applications.

%GENAPP Lite~\cite{genapplite} exposes APIs of GENAPP through SOAP and REST, but this APIfication process is conducted manually.

Our work is related to migrating legacy mainframe applications to a Service-Oriented Architecture. Survey paper on identifying APIs~\cite{survey} does not cover mainframe technology related artifacts and source code constructs in detail. They list only transactions and documentation as a way of identifying APIs. Existing literature proposes creating one API per COBOL program~\cite{mzfm19} and using machine learning techniques to identify APIs~\cite{taammeg22}. Furthermore, the literature~\cite{empcs14} describes how a REST service layer can be created when provided with code and request/response data.

Several modernization tools~\cite{kalia2021mono2micro, icws22, icse23, desai2021graph, ass14, zhw06} discuss application refactoring and communication between candidate microservices, specifically for Java-based applications. However, their concepts and methods cannot be directly applied to COBOL workflows, as COBOL lacks a clear function concept for easy understanding of boundaries and parameters. Additionally, specific COBOL concepts like copybook and CICS link are proprietary to the language. Works like \cite{k10, s01, se96, sdbac21, dsb20} focus on identifying business rules but do not address the identification of API signatures. Recent COBOL-to-Java approaches~\cite{chakravarthy2025c2j, chakravarthy2025saner} focus on
language transformation, whereas we generate API signatures and
standalone compilable COBOL modules.

Screen scraping~\cite{scraping, cfft06} replaces the screen with HTML pages that call APIs exposed by the mainframe application. In contrast to our approach, this approach only exposes screen-interfacing functionalities and not the internal functions of the application.% Moreover, screen scraping can lead to APIs with unnecessary code execution, such as preparatory processes, validation checks, and formatting. In contrast, our proposed methodology aims to extract only the relevant code for the API's functionality or identify and conditionally exclude unnecessary code from execution.

Tools commonly used in the industry, such as IBM Application Discovery and Delivery Intelligence for IBM Z (ADDI)~\cite{addi}, can determine whether a variable has been read or written within a program statement. However, these tools do not identify APIs and do not perform static analyses for establishing requests and responses for a set of statements that are exposed as an API. %To achieve this, a more advanced process (Section~\ref{sec:sign}) involving static analysis of the sequence of read and write operations within the statements, as well as the statements that are called from there recursively, is required.

Our approach complements legacy slicing-based analyses~\cite{zhang2006extracting, hans2025testing}, which compute dependence sets but do not generate compilable modules or API signatures. We instead integrate slicing-like backward and forward data-flow traversal with COBOL-specific constructs (REDEFINES, COPYBOOKs) for executable modularization.

Our methodology is thus significant in its ability to automatically identify APIs using the comprehensive artifacts and compute their signatures. Additionally, we propose code refactoring techniques to enable communication and generate efficient APIs.

%Work~\cite{cfft06} uses screens to identify APIs, treats the code as a black box, and creates wrappers on the legacy system to make them accessible to the web services.

\section{Conclusions}
\label{sec:conclusions}
%With the ability to transform mainframe applications into modern APIs, organizations can enhance system efficiency, integrate with modern technologies, improve security, and achieve greater business agility. However, identifying candidate APIs in mainframe applications is a challenging task due to various factors, such as complex data structures, monolithic programming styles, and esoteric program and variable names. 

In this paper, we propose a comprehensive framework to automate APIfication of mainframe applications. Our approach leverages static program analysis techniques to accurately extract APIs and compute their signatures. 
%Automating this process streamlines the workflow and decreases the amount of manual effort needed.
We validate the correctness of our techniques through construction. We demonstrate the effectiveness and scalability of our approach through experiments and qualitiative study on a public mainframe application (GENAPP) and two industry mainframe applications. The results show that our automated approach significantly reduces the time required for API identification and computation of API signature compared to manual methods. Additionally, we successfully execute a subset of our candidate APIs on an IBM Z mainframe system, showcasing the practical applicability of our approach in real-world scenarios.

Static analysis and LLMs serve distinct yet complementary roles: LLMs provide high-level semantic guidance—for example, suggesting modularization boundaries or service decomposition using natural-language and repository-level cues—while static analysis computes exact API signatures and produces compilable, behavior-preserving modules with full traceability. We envision hybrid workflows that combine LLM-generated hypotheses with static-analysis-driven code materialization for scalable, reliable modernization.

Overall, our proposed framework provides a scalable and efficient solution to address the challenges in mainframe APIfication and facilitates the transformation of mainframe applications into modern, microservices-based architectures.

%Mainframe APIfication is the process of modernizing mainframe applications by making them accessible via modern application programming interfaces (APIs) using tools like zOS Connect. %It can significantly improve system efficiency, facilitate integration with modern technologies, enhance security, and improve business agility etc. 
%However, identifying candidate APIs of mainframe applications pose a significant challenge and is comparatively more challenging than non-mainframe applications due to reasons such as esoteric program and variable names, complex data-stores, monolithic programming styles, and others. Additionally, accurate discovery of APIs in a monolithic architecture is a crucial step in transforming to microservices.

%To address these challenges, we propose a comprehensive framework to identify candidate APIs, determine their API signature, and update relevant code to enable communication through common business use cases on a public mainframe application, GENAPP and on a real life customer application. Our proposed approach leverages various static program analysis techniques to automatically identify candidate APIs in mainframe applications. Our qualitative and quantitative experiments demonstrate the effectiveness of our approach, highlighting significant differences in the time required to identify APIs manually versus our automated approach. Furthermore, we validated a subset of our candidate APIs on a Z mainframe system by executing them successfully, demonstrating the practical applicability of our approach.

Future work includes extending path-sensitive analyses and automating API security and error-handling patterns. We will also combine static analysis for guaranteed, structure-preserving refactoring with LLMs for more open-ended transformations.

\section{Future Works}
\label{sec:future}
%While the results of our study are promising and offer important potential applications, it is important to note that this is only a preliminary step towards bridging the significant gap that currently exists in Mainframe APIfication. 
As we look to the future, there are numerous potential areas for further research and development in the field of mainframe APIfication. Some of these avenues include:

\begin{enumerate}
\item Merging APIs: Exploring heuristics and algorithms to identify opportunities for merging multiple APIs into a single, cohesive API. This could optimize the API landscape, reducing redundancy and simplifying integration for developers.

\item Domain-Based Models: Leveraging domain-based models and natural language processing techniques to further enhance the identification of candidate APIs. By aligning program artifacts with domain ontologies, we can improve the accuracy and relevance of the identified APIs.

\item Code Dependencies: Addressing code dependencies and ensuring that the extracted APIs function cohesively within the larger application ecosystem. This involves analyzing and resolving potential conflicts or dependencies that may arise during the APIfication process.

\item Production-Ready APIs: Going beyond the identification process, ensuring the APIs are production-ready by implementing robust error handling mechanisms, addressing data dependencies, thoroughly testing the refactored code, stress test the APIs.

\item API Management: Registering and managing APIs with appropriate tools and platforms, such as IBM API Connect, to govern access, monitor usage, and enforce security policies.
\end{enumerate}

By pursuing these research areas, we can further enhance the efficiency, effectiveness, and scalability of mainframe APIfication, enabling organizations to modernize their mainframe applications while maintaining high standards of performance and reliability.

\bibliographystyle{ACM-Reference-Format}
%\bibliography{sample-base}

\appendix

\end{document}